\PassOptionsToPackage{hyphens}{url}
\RequirePackage{float} 
\documentclass[a4paper,fleqn]{cas-dc}
\usepackage[authoryear]{natbib}
\usepackage{algorithm}
\usepackage{algorithmic}
\usepackage{placeins}
\usepackage{amsthm}
\usepackage{microtype}
\newcommand{\compactalgorithm}{%
  \small
  \setlength{\abovedisplayskip}{4pt plus 1pt minus 1pt}%
  \setlength{\belowdisplayskip}{4pt plus 1pt minus 1pt}%
  \setlength{\abovedisplayshortskip}{2pt plus 1pt minus 1pt}%
  \setlength{\belowdisplayshortskip}{4pt plus 1pt minus 1pt}%
  \setlength{\jot}{2pt}%
}
\ExplSyntaxOn
\NewCommandCopy \LieDiscoverPrintOrcid \printorcid
\RenewDocumentCommand \printorcid { }
  { \seq_if_empty:NF \g_stm_orcid_seq { \LieDiscoverPrintOrcid } }
\ExplSyntaxOff
\newtheorem{theorem}{Theorem}
\newtheorem{lemma}{Lemma}
\makeatletter
\def\@dbflt#1{%
  \@ifnextchar[{\@xdblfloat{#1}}{%
    \edef\reserved@a{\noexpand\@xdblfloat{#1}[\csname fps@#1\endcsname]}%
    \reserved@a}}
\makeatother

\ExplSyntaxOn
\cs_set:Npn \__first_footerline:
  {
    \ifnum\theblind=0\relax
      {\small\sffamily \__short_authors:}
    \fi
  }
\ExplSyntaxOff
\begin{document}
\shorttitle{LieDiscover: Open-form Symmetry Discovery}
\shortauthors{X. Li et~al.}
\title[mode=title]{LieDiscover: Adaptive Symbolic Library Construction for Explicit Open-form Symmetry Discovery}

\author[1]{Xinxin Li}[
                        auid=000,bioid=1,
                        orcid=0009-0008-8962-0084]

\ead{51265500102@stu.ecnu.edu.cn}
\author[2]{Jianming Ma}
\ead{nymjm11@sjtu.edu.cn}
\author[4,6]{Xingyu Cui}
\ead{cuixingyu212@163.com}

\author[5]{Da Li}
\ead{dli@nenu.edu.cn}

\author[3,6]{Juan Zhang}
\cormark[1]
\ead{zhang_juan@buaa.edu.cn}

\author[4,6]{Junping Yin}
\cormark[1]
\ead{yinjp829829@126.com}
\cortext[cor1]{Corresponding authors.}

\affiliation[1]{organization={School of Mathematical Sciences, East China Normal University},
                city={Shanghai},
                postcode={200241}, 
                country=China}
\affiliation[2]{organization={School of Computer Science, Shanghai Jiao Tong University}, city={Shanghai},
                postcode={200240}, 
                country=China}     
\affiliation[3]{organization={Institute of Artificial Intelligence, Beihang University},
                postcode={100191}, 
                city={Beijing},
                country={China}}           
\affiliation[4]{organization={Institute of Applied Physics and Computational Mathematics},
                postcode={100094}, 
                city={Beijing},
                country={China}}
\affiliation[5]{organization={Academy for Advanced Interdisciplinary Studies, Northeast Normal University},
                postcode={130024}, 
                city={Changchun},
                state={Jilin},
                country={China}}
\affiliation[6]{organization={Shanghai Zhangjiang Institute of Mathematics},
                postcode={201203}, 
                city={Shanghai},
                country={China}}

\begin{abstract}
Discovering symmetries from observational data provides interpretable insights into dynamical systems and priors for scientific machine learning. However, existing symmetry discovery methods face two limitations. Many methods represent symmetries implicitly, making it difficult to obtain explicit infinitesimal generators that can serve as physical priors for learning dynamical systems. Meanwhile, methods that explicitly identify generators often rely on predefined function libraries, limiting their ability to capture complex symmetries involving high-order polynomials or transcendental functions. To address these limitations, we propose LieDiscover, a framework that explicitly identifies infinitesimal generators through an iterative symbolic search process that dynamically expands the symbolic function library to capture complex functional forms. Specifically, at each search step, we design an encoder-decoder architecture to generate a candidate symbolic expression. An adaptive singular value decomposition (SVD)-based evaluation is introduced to assess how each candidate contributes to the recoverable symmetry subspace. The resulting evaluations are incorporated into step-wise rewards to guide reinforcement learning toward discovering richer symmetry structures with concise symbolic representations. Experiments on eight partial differential equation (PDE) benchmarks demonstrate that LieDiscover recovers explicit generators involving higher-order polynomial, exponential, and trigonometric terms, while achieving more accurate symmetry subspace recovery and improved robustness to noise. Furthermore, the recovered generators enable effective symmetry-based data augmentation, improving the accuracy of downstream neural PDE solvers.
\end{abstract}

\begin{keywords}
Symmetry discovery \sep Reinforcement learning \sep PDE solving
\end{keywords}

\maketitle

\section{Introduction}

\label{sec:Introduction}
Symmetry is fundamentally defined as invariance under a group of transformations \citep{weyl2015symmetry}. It is widely applied across fields ranging from fundamental physics \citep{gross1996role} to modern machine learning \citep{bronstein2021geometric,cohen2016group,kondor2018generalization}. Crucially, leveraging symmetries has been shown to reduce problem complexity \citep{udrescu2020ai}, improve model generalization \citep{benton2020learning}, and provide useful priors for physical modeling \citep{brandstetter2022lie,akhound2023lie}. These benefits typically rely on prior knowledge of the underlying symmetries, which is often unavailable in practice. Consequently, discovering symmetries directly from data has become an increasingly important problem.

Existing symmetry discovery methods have evolved from selecting symmetries within predefined transformation families to discovering infinitesimal generators directly from data. While methods such as LieGAN \citep{yang2023generative} and LieSD \citep{hu2025symmetry} have successfully automated the discovery of linear symmetries, their underlying mathematical modeling relies on matrix representations. This limitation makes their methods ineffective at capturing nonlinear symmetries in real-world systems. To improve the ability to model complex symmetries, recent methods have extended their capabilities to learn nonlinear symmetries. LaLiGAN \citep{yang2023latent} overcomes this challenge by using an encoder-decoder architecture to map data into a latent space, where nonlinear symmetries can be modeled as linear transformations. Alternatively, \citet{ko2024learning} parameterize the infinitesimal generators using multilayer perceptrons (MLPs) to extract symmetries directly from data.

Despite these advancements in capturing nonlinear symmetries, existing methods still struggle to meet the demands of downstream physical applications. Specifically, approaches built upon black-box neural architectures suffer from a lack of interpretability, failing to output intuitive formulas for generators. Other methods provide explicit representations. For example, \citet{shaw2024symmetry} express vector fields as linear combinations of prescribed features, mainly polynomials. Their expressiveness is limited by the chosen features. More importantly, many existing approaches to nonlinear symmetry discovery struggle to estimate the dimension of the recoverable symmetry subspace and provide explicit generator expressions. The lack of explicit expressions makes it difficult to construct symmetry transformations for data augmentation in PDE solvers. LieNLSD \citep{hu2025explicit} formulates infinitesimal generators as linear combinations of a predefined function library, effectively reducing the complex symmetry discovery task to a linear-system solving problem. It estimates the dimension of the recoverable symmetry subspace within its predefined function library and provides explicit generator expressions. However, its predefined library limits recovery to generators within the span of the specified functions.

To overcome this restriction, we propose LieDiscover, a novel framework that does not require the complete candidate library to be predefined in advance. By jointly discovering a function library and its corresponding coefficients, our method enables the automated discovery of open-form infinitesimal generators. Here, {open-form} refers to dynamically constructing library expressions by composing variables and operators. Although the search is bounded by expression complexity constraints, the candidate expressions are not explicitly predefined in advance. Initialized with a basic library of constant and linear functions, our method leverages an encoder-decoder architecture to dynamically generate and append symbolic trees to the library. Guided by immediate step-wise rewards, this entire generative process can be efficiently optimized using reinforcement learning. In summary, our contributions are as follows: 
\begin{itemize}
    \item We formulate open-form infinitesimal generator discovery as a joint optimization problem over the function library and its coefficients, enabling the automated discovery of explicit open-form generators.
    \item We establish a reinforcement learning framework for adaptive symbolic library construction. Immediate step-wise rewards guide the search toward concise expressions that expand the recovered symmetry space.
    \item We develop an adaptive SVD method to estimate the dimension of the recoverable symmetry subspace within the constructed symbolic library. It proposes candidate dimensions from spectral gaps and selects among them using residual and stability checks.
\end{itemize}

\section{Related Work}
\label{sec:Related Work}
\subsection{Symmetry Discovery}
Symmetry discovery aims to identify transformations that preserve the underlying structure of data or governing dynamics. Early methods \citep{benton2020learning,zhou2020meta,romero2022learning,van2023learning,krippendorf2021detecting} perform symmetry selection within a predefined transformation family. These methods typically parameterize candidate transformations using known group actions or structured transformation models, and learn the transformations that best explain the observed data. Although they can automatically identify symmetries without requiring their exact parameters to be specified, their discovery capabilities remain constrained by the predefined transformation families and parameterizations.

Recent methods discover symmetries through infinitesimal generators, which provide a natural representation of continuous transformations. These methods can be divided into model-based and distribution-based approaches. Model-based approaches identify symmetries by exploiting the equivariance or invariance properties of neural networks. Within this line, L-conv \citep{dehmamy2021automatic} incorporates learnable Lie algebra generators into neural network architectures and jointly optimizes the generators and network parameters, enabling automatic symmetry discovery during model training. In contrast, LieGG \citep{moskalev2022liegg} and LieSD \citep{hu2025symmetry} identify symmetries from trained neural networks. Specifically, LieGG uses a polarization matrix to extract invariant transformations from learned input-output mappings, while LieSD constructs a linear system from network inputs, outputs, and gradients, and recovers Lie algebra bases through null-space analysis. Another line of work identifies symmetries by exploiting the invariance of data distributions. SymmetryGAN \citep{desai2022symmetry} and LieGAN \citep{yang2023generative} formulate symmetry discovery as an adversarial learning problem, in which a generator learns transformations that preserve the data distribution against a discriminator. In particular, LieGAN parameterizes continuous transformations using Lie algebra bases, enabling the discovery of interpretable continuous symmetries. Despite their different discovery mechanisms, these methods primarily focus on linear group actions, with infinitesimal generators represented by constant matrices, limiting their ability to capture general nonlinear symmetries.

More recent studies have begun to explore nonlinear symmetry discovery. LaLiGAN \citep{yang2023latent} learns a latent representation in which nonlinear group actions become linear, enabling symmetry discovery through a linear representation in the latent space. This approach provides a way to handle nonlinear transformations while retaining the advantages of linear symmetry learning, although it relies on the learned latent-space representation. Rather than using latent-space linearization, \citet{ko2024learning} parameterize nonlinear infinitesimal generators with neural networks. This representation supports flexible nonlinear transformations but does not directly provide explicit symbolic expressions. In contrast, \citet{shaw2024symmetry} estimate tangent vector fields as linear combinations of prescribed features, mainly polynomials. Their method provides explicit generator expressions, but the chosen features restrict the forms it can represent.

Another direction is to discover explicit nonlinear generators through symbolic methods. For ODE systems with known governing equations, \citet{kahlmeyer2025discovering} employ search-based symbolic regression to identify Lie point symmetry generators from the determining equations. This approach demonstrates the potential of symbolic search for discovering complex generator expressions, but requires access to the governing equations. In contrast, LieNLSD \citep{hu2025explicit} addresses symmetry discovery from observational data without requiring prior knowledge of the governing equations. It represents nonlinear infinitesimal generators using a predefined function library and formulates their identification as a linear system. Specifically, it constructs a linear system based on the infinitesimal symmetry criterion and recovers the coefficients of candidate functions through null-space analysis. However, its reliance on a predefined library restricts the learned generators to the span of the specified functions. To overcome this restriction, we propose LieDiscover, an explicit framework for open-form symmetry discovery. Instead of selecting candidate functions from a fixed library, LieDiscover dynamically constructs mathematical expressions by composing variables and operators, enabling the discovery of explicit nonlinear generators beyond a fixed candidate space.

\subsection{Applications of Symmetry}
Symmetries have been widely exploited in machine learning, particularly in the design of equivariant networks \citep{cohen2016group,fuchs2020se} and PDE solving \citep{akhound2023lie,alpar2024applications}. In equivariant networks, symmetry constraints are incorporated into model architectures to improve generalization by exploiting known transformation properties. In neural PDE solvers, symmetries provide structural information about governing equations and their solution spaces, which can be used to improve model training and generalization.

Recent studies have incorporated PDE symmetries into physics-informed neural networks (PINNs) through symmetry-aware loss functions \citep{akhound2023lie,arora2024invariant} and data augmentation \citep{brandstetter2022lie,ko2024learning}. The former introduces symmetry constraints into training objectives, while the latter applies symmetry transformations to existing solution data to generate additional training samples that preserve the underlying PDE structure. Symmetry-aware PINNs and LPSDA typically use symmetries derived from known governing equations. In contrast, \citet{ko2024learning} learn infinitesimal generators directly from data and use the resulting transformations for augmentation. LieDiscover complements these approaches by discovering explicit symbolic generators from data, which can then be used for PDE data augmentation.

\begin{figure*}[pos=tb]
    \centering
    \includegraphics[width=0.85\linewidth]{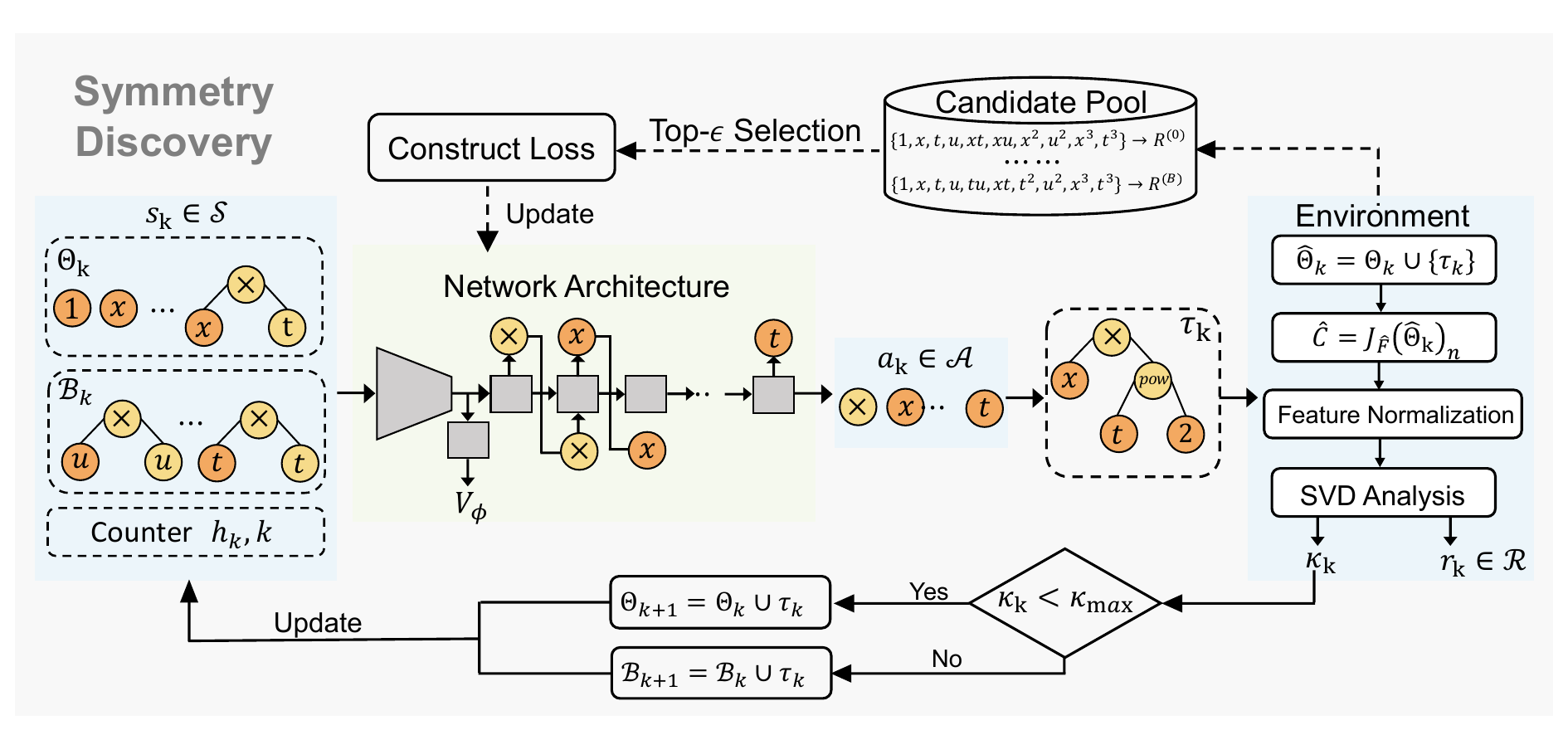}
    \caption{The framework of LieDiscover. Symmetry discovery is formulated as an MDP, using solid arrows for information flow and dashed arrows for network training.}
    \label{fig:liediscover}
\end{figure*}

\section{Preliminaries}
\label{sec:Preliminary}
Continuous symmetries in dynamical systems are transformations that leave the system invariant. Since searching for these global transformations directly is challenging, modern methods instead search for their infinitesimal generators \citep{dehmamy2021automatic,hu2025explicit,hu2025symmetry}. 

For a Lie group acting on the space $X \times U = \mathbb{R}^p \times \mathbb{R}^q$, its infinitesimal generator $\mathbf{v}$ and its $n$-th prolongation $\text{pr}^{(n)}\mathbf{v}$ can be represented as \citep{hu2025explicit}:
\begin{equation}
\label{eq:infinitesimal generator}
\begin{aligned}
\mathbf{v} &= W\Theta(x, u) \cdot \nabla,\\
\text{pr}^{(n)}\mathbf{v} &= \Theta_n(x, u^{(n)})\operatorname{vec}(W) \cdot \nabla,
\end{aligned}
\end{equation}
where $W \in \mathbb{R}^{(p+q) \times r}$ is a coefficient matrix, $\Theta(x,\allowbreak  u) \in \mathbb{R}^{r \times 1}$ is a predefined function library, and $\Theta_n$ is the prolonged matrix determined by $\Theta$. Throughout, $\operatorname{vec}(W)$ stacks the rows of $W$. Crucially, the prolongation maintains linearity with respect to $W$. This linearity transforms the classical infinitesimal criterion into a straightforward linear algebraic problem \citep{hu2025explicit}:
\begin{theorem}
\label{thm:linearized_symmetry}
Let $\mathcal{E}: F(x,\allowbreak u^{(n)})=0$ be a smooth, locally solvable differential equation system, with $J_F$ of full row rank on the equation manifold $\mathcal{M}=\{(x,u^{(n)}):F(x,u^{(n)})=0\}$ in jet space. Local solvability means that each point of $\mathcal M$ lies on the prolonged graph of a local solution. Under these nondegeneracy conditions \citep[Sections~2.3 and~2.6]{olver1993applications} and the parameterization of Eq.~\eqref{eq:infinitesimal generator}, a smooth vector field $\mathbf{v}$ generates a local one-parameter symmetry group of $\mathcal{E}$ if and only if
\begin{equation}
\label{eq:linear_criterion}
J_F(x,u^{(n)}) \Theta_n(x,u^{(n)}) \operatorname{vec}(W) = 0, \quad \forall (x,u^{(n)})\in\mathcal{M},
\end{equation}
where $J_F$ denotes the Jacobian matrix of $F$.
\end{theorem}

On discrete empirical data sampled from $\mathcal{M}$, Equation \eqref{eq:linear_criterion} reduces to a homogeneous linear system $C\operatorname{vec}(W)=0$. The numerical null space of $C$ provides candidate coefficient vectors $\operatorname{vec}(W)$ and an estimate of the symmetry dimension supported by the selected library. Finite sampling can introduce spurious null directions, while linear identities among library functions can produce nonzero coefficients representing the zero vector field. We therefore validate candidate subspaces and remove verified library identities before extracting independent generators. Here, the number of generators refers to the dimension of the recovered subspace within the selected library.

However, applying Theorem~\ref{thm:linearized_symmetry} with a
fixed library $\Theta$ restricts the search to the span of that library. To address this limitation, we propose LieDiscover, an open-form generator discovery framework. Instead of fixing the complete library in advance, this framework dynamically constructs mathematical expressions by composing selected variables and operators. The search is conducted within a prescribed budget to control expression complexity.

\section{Methodology}
\label{sec:Methodology}
LieDiscover discovers open-form symmetries through sequential symbolic library construction. The overall methodology consists of three main stages: Jacobian computation, symmetry discovery, and PDE Application. First, in Section~\ref{sec:jacobian}, we identify an explicit governing relation from observational data and cache its Jacobians, providing the necessary information to evaluate candidate symmetry generators using Theorem~\ref{thm:linearized_symmetry}. Next, we formulate symbolic library construction as a Markov decision process (MDP) in Section~\ref{sec:MDP Definition}, and introduce the actor--critic architecture in Section~\ref{sec:Network Architecture} to guide the sequential generation of mathematical expressions. Section~\ref{sec:Training Method} describes the policy optimization method, followed by sparse basis recovery and library pruning to obtain a compact set of independent symmetry generators. The complete discovery pipeline is summarized in Section~\ref{sec:overall_algorithm}, while Fig.~\ref{fig:liediscover} illustrates the core framework for symmetry discovery. Finally, in Section~\ref{sec:PDE Application}, we demonstrate how the discovered symmetry generators can be used to augment training data for downstream PDE solving tasks.

\subsection{Jacobian Computation}
\label{sec:jacobian}
\begin{figure}[pos=t]
    \centering
    \includegraphics[width=0.75\linewidth]{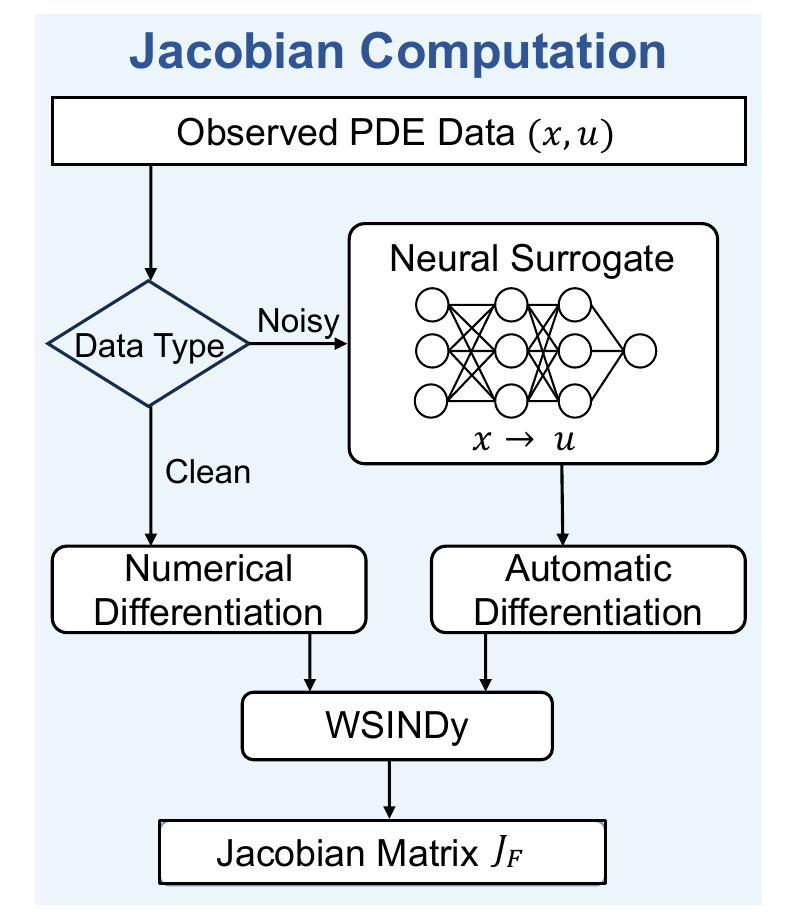}
    \caption{Jacobian computation before symbolic library search.}
    \label{fig:data}
\end{figure}

The symmetry criterion in Eq.~\eqref{eq:linear_criterion} requires the Jacobian $J_F$ of the governing equation, which is generally unknown when only observational data are available. Therefore, before symbolic library search, we first reconstruct the required PDE data, identify an explicit governing equation, and compute its Jacobian. This preparation consists of data processing, equation identification, and Jacobian computation, as illustrated in Fig.~\ref{fig:data}.

For PDE data, let $x=(x^1,\ldots,x^p)$ denote independent variables, $u=(u^1,\ldots,u^q)$ dependent fields, and $u^{(n)}$ their derivatives up to order $n$. We use a multi-index $\boldsymbol{\beta}_j=(\beta_{j1},\ldots,\beta_{jp})\in\mathbb{N}_0^p$ to specify a partial differential operator:
\[
D_{\boldsymbol{\beta}_j}
=\frac{\partial^{|\boldsymbol{\beta}_j|}}
{\partial (x^1)^{\beta_{j1}}\cdots
\partial (x^p)^{\beta_{jp}}},
\qquad
|\boldsymbol{\beta}_j|=\sum_{k=1}^{p}\beta_{jk}.
\]
Given observed data $\{(x_i,u_i)\}_{i=1}^{N_{\mathcal D}}$, we first obtain the required field values and derivatives according to the data quality. For clean data, we directly use the observed grid values and numerical differentiation to compute the required derivatives. For noisy data, we first fit a fully connected coordinate network $\widehat u_\omega(x)$ to reconstruct each trajectory, and then use automatic differentiation to obtain the required derivatives, including mixed derivatives.

Based on the processed data, we construct the regression inputs
\begin{equation}
\label{eq:surrogate_inputs}
    y_i=D_{\boldsymbol\beta_0}u_i,\quad
    z_i=(x_i,D_{\boldsymbol\beta_1}u_i,\ldots,D_{\boldsymbol\beta_s}u_i),
\end{equation}
where $y_i$ contains the selected temporal derivatives. The remaining inputs are state variables and spatial derivatives. We then identify the governing equation using WSINDy \citep{messenger2021weak}. Specifically, let $\boldsymbol{\psi}(z)=(\psi_1(z),\ldots,\psi_M(z))^\top$ denote a library of candidate PDE terms. By integrating the candidate terms and temporal targets against compactly supported test functions, WSINDy constructs a weak regression system
\begin{equation}
\label{eq:weak_regression}
    G\xi \approx B,
\end{equation}
where $G$ contains the weak candidate features, $B$ contains the weak temporal targets, and $\xi\in\mathbb{R}^{M\times q}$ is the coefficient matrix. We estimate a sparse coefficient matrix $\widehat\xi$ using sequential thresholded least squares, with the threshold selected on validation data. The resulting explicit governing equation is
\begin{equation}
\label{eq:identified_equation}
    \widehat f(z)=\widehat\xi^\top\boldsymbol{\psi}(z).
\end{equation}

With the identified equation, we define the residual and compute its Jacobian at each reference sample:
\begin{equation}
\label{eq:identified_jacobian}
    \widehat F(z_i,y_i)
        =\widehat f(z_i)-y_i,\quad
    J_{\widehat F,i}
        =\left[
            \widehat\xi^\top J_{\boldsymbol{\psi}}(z_i),\,-I_q
          \right].
\end{equation}
$J_{\boldsymbol{\psi}}$ denotes the Jacobian of the candidate feature vector with respect to its input variables. We compute these Jacobians by automatic differentiation of the identified equation. The identified equation is then kept fixed during policy training, and the Jacobians at the reference samples are computed and cached when the search environment is initialized. This allows all candidate symmetry libraries to be evaluated using the same equation and Jacobian information without repeated computation.


\subsection{MDP Definition}
\label{sec:MDP Definition}
We formulate symbolic library construction as a finite-horizon Markov decision process (MDP) $\langle\mathcal{S},\allowbreak\mathcal{A},\allowbreak\rho_0,\allowbreak\mathcal{T},\allowbreak r,\allowbreak K\rangle$, where $\mathcal{S}$ is the state space, $\mathcal{A}$ is the action space, $\rho_0$ is the initial state distribution, $\mathcal{T}$ is the transition function, $r$ is the reward function, and $K$ is the number of search steps per episode. At each step, the agent proposes one expression, and the environment evaluates the expanded library, accepts or rejects the proposal, and assigns a reward based on the change in the estimated symmetry dimension and the expression complexity. 

\paragraph{\textbf{Initial Distribution.}}
The search begins with a small library $\Theta_0$ of basic expressions, such as $\{1,t,x,u\}$ for the Burgers' equation. At this point, no candidates have been rejected and no search steps have been taken. We define this starting configuration as the initial MDP state:
\begin{equation}
\label{eq:initial_state}
    s_0=(\Theta_0,\varnothing,0,0).
\end{equation}
For a given task, every episode starts from this same state. The initialization is therefore deterministic, with initial distribution $\rho_0=\delta_{s_0}$, where $\delta_{s_0}$ assigns probability one to $s_0$.

\paragraph{\textbf{State Space}.}
At search step $k=0,\dots,K-1$, the state is
\begin{equation}
\label{eq:state_space}
s_k=(\Theta_k,\mathcal B_k,h_k,k)\in\mathcal{S},
\end{equation}
where $\Theta_k$ contains accepted expressions and $\mathcal B_k$ records rejected expressions. Keeping the rejected expressions lets the policy use previous failures to guide new proposals and helps it avoid repeating unsuccessful choices. The counter $h_k$ tracks how many consecutive steps have passed without increasing the estimated symmetry dimension. The counter $k$ records how many proposals have been made and determines how many search steps remain. 

\paragraph{\textbf{Action Space.}}
We define an MDP action $a_k\in\mathcal{A}$ as a symbolic expression. At each step, the agent proposes one such expression as a candidate for expanding $\Theta_k$. To obtain this expression, the policy generates a sequence of tokens, which is then parsed into an expression tree. Details of token generation and the constraints that ensure valid expressions are provided in Section~\ref{sec:Network Architecture}.

\paragraph{\textbf{State Transition.}}
Given $a_k$, the environment first forms the trial library
\begin{equation}
\label{eq:trial_library}
    \widetilde\Theta_k=\Theta_k\cup\{a_k\}.
\end{equation}
To evaluate this trial library, the environment constructs its prolonged features and combines them with the cached Jacobians to form the linear system $C(\widetilde\Theta_k)\operatorname{vec}(W)=0$.
Algorithms~\ref{alg:theta_n_construction} and~\ref{alg:recursive_coeff} in Appendix \ref{app:algorithms} describe the construction of the prolonged features.

Before search, we fix a reference pool $\mathcal I_0$, three random search subsets $\mathcal I_b$ ($b=1,2,3$), and a validation set $\mathcal I_v$. The search subsets are drawn from the fitting data and exclude the validation samples. We denote these sets collectively by $\mathcal I=(\mathcal I_0,\mathcal I_1,\mathcal I_2,\mathcal I_3,\mathcal I_v)$. The reference pool is used for feature scaling and generator evaluation, the search subsets for dimension estimation and stability assessment, and the validation set for residual evaluation.

Symbolic features and their prolonged derivatives may have substantially different scales, which can affect numerical dimension estimation. To improve numerical conditioning, we normalize the prolonged features using scales computed from the reference pool $\mathcal I_0$. Specifically, let $B_i=\Theta_n(x_i,u_i^{(n)})$ denote the prolonged matrix of the trial library at sample $i$. For each coefficient column $j$, we compute its root-mean-square scale over the $M=|\mathcal I_0|$ reference samples:
\begin{equation}
\label{eq:prolonged_column_scale}
\nu_j
=\left(\frac{1}{M}\sum_{i\in\mathcal I_0}
\|B_i[:,j]\|_2^2\right)^{1/2}.
\end{equation}
If $\nu_j\le\varepsilon_s$, we set $\nu_j=1$, where $\varepsilon_s>0$ is a small threshold. We then define $S_k=\operatorname{diag}(\nu_1,\ldots,\nu_{r(p+q)})$, with $r=|\widetilde\Theta_k|$, and obtain the normalized prolonged matrices:
\begin{equation}
\label{eq:normalized_prolonged_matrix}
\bar B_i=B_iS_k^{-1}.
\end{equation}
The scales are computed once on the reference pool and consistently applied to all sample sets. To preserve the represented generator values, the normalized coefficients satisfy $\operatorname{vec}(\bar W)=S_k\operatorname{vec}(W)$. Accordingly, the normalized symmetry system is
\begin{equation}
\label{eq:normalized_constraint_system}
\bar C_k\operatorname{vec}(\bar W)=0,\quad
\bar C_k=C(\widetilde\Theta_k)S_k^{-1}.
\end{equation}
The original coefficients are recovered as
\begin{equation}
\label{eq:original_coefficient_recovery}
\operatorname{vec}(W)=S_k^{-1}\operatorname{vec}(\bar W).
\end{equation}
Since $S_k$ is invertible, normalization preserves the represented generator values and the exact nullity of the symmetry system.

We apply SVD to the normalized symmetry systems from the three search subsets. Since the singular value spectrum may contain multiple plausible gaps, we select the three largest positive log gaps from each subset to propose candidate dimensions. Their union is denoted by $\mathcal N$.

Each candidate $d\in\mathcal N$ is evaluated using four checks. First, $e_0(d)$ and $e_R(d)$ measure the symmetry residuals on the reference pool and validation set, respectively, assessing residual accuracy and generalization. Second, $e_P(d)$ and $e_X(d)$ measure the consistency of the normalized coefficient subspaces and the corresponding generator-function spaces across the three search subsets. Both checks use basis-invariant principal angles \citep{drineas2019low,park2023fast}. While $e_P(d)$ assesses coefficient-space stability, $e_X(d)$ evaluates whether the resulting generator functions are also consistent. The latter uses the same reference points and helps detect unstable cancellations between correlated features.

We accept a candidate only if it passes all four checks and select the largest admissible dimension:
\begin{equation}
\label{eq:validated_dimension}
\begin{aligned}
    &\mathcal N_a=\{d\in\mathcal N:\max(e_0(d),e_R(d))\le\varepsilon_R,\\
    &\hspace{4.5em}e_P(d)\le\varepsilon_P,\ e_X(d)\le\varepsilon_X\},\\
    &d(\widetilde\Theta_k)=\max\mathcal N_a.
\end{aligned}
\end{equation}
We use the same empirical tolerances for all benchmarks: $\varepsilon_R=10^{-3}$ and $\varepsilon_P=\varepsilon_X=10^{-2}$. If a candidate loses a direction when evaluated as generator functions, we set $e_X(d)=+\infty$. If no candidate passes, we return $d=\bot$ and $\kappa=+\infty$. Otherwise, $\kappa$ is the largest condition number of the retained non-null spectra across the three subsets. Further details of these checks are provided in Appendix~\ref{app:adaptive_svd}, and the complete procedure is summarized in Algorithm~\ref{alg:adaptive_svd_filtering}.

The candidate is accepted if the condition number does not exceed $\kappa_{\max}$, and is otherwise added to the rejection set:
\begin{equation}
\label{eq:library_transition}
(\Theta_{k+1},\mathcal B_{k+1})=
\begin{cases}
(\widetilde\Theta_k,\mathcal B_k),
&\kappa(\widetilde\Theta_k)\le\kappa_{\max},\\
(\Theta_k,\mathcal B_k\cup\{a_k\}),
&\text{otherwise}.
\end{cases}
\end{equation}
An accepted expression remains in the library even if the estimated symmetry dimension does not increase, allowing it to be combined with later additions. Let $d_k=d(\Theta_k)$ and $d_{k+1}=d(\Theta_{k+1})$ denote the estimated dimensions before and after this step. If the dimension increases, we reset the counter to zero. Otherwise, we increase it by one:
\begin{equation}
\label{eq:delay_update}
\begin{aligned}
    \Delta d_k&=d_{k+1}-d_k,\\
    h_{k+1}&=\begin{cases}
        0,&\Delta d_k>0,\\
        h_k+1,&\Delta d_k\le0.
    \end{cases}
\end{aligned}
\end{equation}
Every proposal advances the search counter from $k$ to $k+1$, including a rejection. Together, these updates specify $s_{k+1}=\mathcal T(s_k,a_k)$.

\paragraph{\textbf{Reward Function.}}
The reward encourages increases in the estimated symmetry dimension while favoring simple expressions:
\begin{equation}
\label{eq:reward equation}
r_k =
\begin{cases}
\eta^{h_k}\Delta d_k+\alpha e^{-\lambda c(a_k)},
&a_k\text{ is accepted},\\
0,&a_k\text{ is rejected},
\end{cases}
\end{equation}
where $c(a_k)$ is the number of tokens in the expression. The first term measures the dimension change, with $0<\eta\le1$ controlling its decay with waiting time. When $\eta<1$, the same positive gain receives a larger reward if found sooner. The second term rewards shorter expressions: $\alpha>0$ sets its weight, and $\lambda>0$ controls how quickly it decreases with expression length. An accepted expression with no dimension change receives only this simplicity bonus. If adaptive SVD yields a dimension decrease, the first term contributes negatively. Appendix~\ref{app:proof} gives a sufficient condition for larger positive dimension gains to receive higher rewards across different delays and expression lengths.

\paragraph{\textbf{Termination.}}
Each episode terminates after $K$ proposals and returns the accepted library $\Theta_{K}$. Every proposal counts as one step, whether accepted or rejected. A lack of dimension improvement does not trigger early termination. Library selection across episodes and generator extraction are described in Section~\ref{sec:Training Method}.

\subsection{Network Architecture}
\label{sec:Network Architecture}
To generate actions and estimate returns for the MDP, the actor and critic share an encoded representation of the symbolic state. The actor is an encoder-decoder model that maps the current accepted and rejected libraries to a valid expression tree, while the critic maps the same state representation to a scalar value. Fig. ~\ref{fig:ac} illustrates the architectures of the actor and critic.
\begin{figure}[pos=t]
    \centering
    \includegraphics[width=\linewidth]{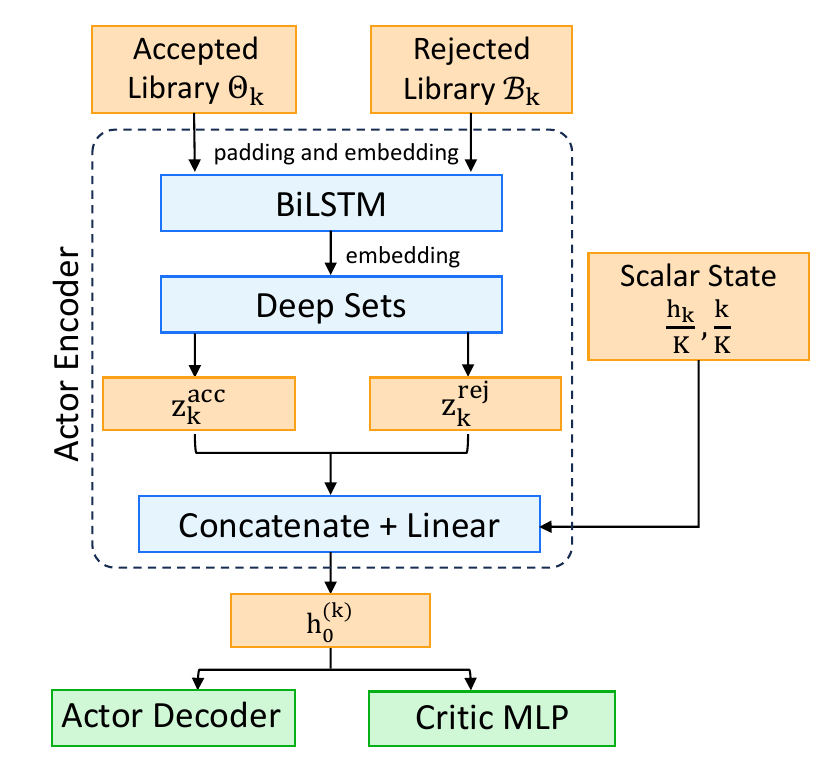}
    \caption{The architecture of the actor and critic.}
    \label{fig:ac}
\end{figure}

\noindent\textbf{Actor.}
The actor models $\pi_\theta(a_k\mid s_k)$ using an encoder-decoder architecture. It encodes the current search state $s_k$ and autoregressively generates a candidate symbolic expression $a_k$ in preorder form.

The encoder maps the state $s_k$ into a fixed-dimensional context vector. Symbolic trees are padded to the longest token sequence in the current state, while variable-size libraries are handled by the set aggregation below. The tokens are first embedded and then processed by a Bidirectional LSTM \citep{graves2005framewise}. We concatenate the final forward and backward hidden states of the BiLSTM to obtain the tree embedding $e_i$. Since the order of expressions in a library is irrelevant, we use Deep Sets \citep{zaheer2017deep} to encode the accepted and rejected libraries separately:
\begin{equation}
\label{eq:global representations}
\begin{aligned}
\boldsymbol{z}_k^{\mathrm{acc}} &= \frac{1}{|\Theta_k|}\sum_{a_i\in\Theta_k}\mathrm{MLP}(e_i),\\
\boldsymbol{z}_k^{\mathrm{rej}} &= \frac{1}{|\mathcal B_k|}\sum_{a_i\in\mathcal B_k}\mathrm{MLP}(e_i).
\end{aligned}
\end{equation}
If the rejection set is empty, we set $\boldsymbol{z}_k^{\mathrm{rej}}=\boldsymbol 0$. We then normalize the scalar components $h_k$ and $k$ by $K$ and concatenate them with the two library representations. A linear layer maps the resulting vector to the context vector:
\begin{equation}
h_0^{(k)} = \mathrm{Linear}\left(
\left[
\boldsymbol{z}_k^{\mathrm{acc}};
\boldsymbol{z}_k^{\mathrm{rej}};
\frac{h_k}{K};
\frac{k}{K}
\right]
\right).
\end{equation}

The decoder initializes its LSTM hidden state with $h_0^{(k)}$ and generates a preorder token sequence $a_k=[y_1,\dots,y_l]$ autoregressively \citep{petersen2019deep}. To ensure valid symbolic trees, it tracks the number of unfilled child slots:
\begin{equation*}
\label{eq:arity tracking}
b_l = b_{l-1} - 1 + \operatorname{arity}(y_l),
\end{equation*}
where $b_0=1$ and $\operatorname{arity}(y_l)$ denotes the arity of token $y_l$. Generation stops when $b_l=0$. At each step, a dynamic mask $M_l\in\{0,1\}^{|\mathcal{V}|}$ excludes invalid tokens and choices forbidden by prior constraints before the softmax operation.

\noindent\textbf{Critic.}
The critic estimates the value of the current symbolic state using the shared context vector $h^{(k)}_0$. A lightweight MLP outputs
\begin{equation}
    \label{eq:critic}
    V_\phi(s_k)=\mathrm{MLP}_\phi(h_0^{(k)}),
\end{equation}
which predicts return-to-go from the elite trajectories selected for critic training. It serves as a learned baseline in the actor update. Because training uses selected trajectories, it need not equal the unconditional policy value $V^{\pi_\theta}(s_k)$.

\subsection{Policy Optimization and Basis Extraction}
\label{sec:Training Method}
Our search objective is to maximize the expected cumulative reward,
\begin{equation}
\label{eq:optim_problem}
\pi^*
=
\arg\max_{\pi_\theta}
\mathbb{E}_{\pi_\theta}
\left[
\sum_{k=0}^{K-1} r(s_k,a_k)
\right].
\end{equation}
Since successful symbolic expressions are sparse in the open-form search space, we use a heuristic actor--critic update with success-based actor selection and risk-seeking \citep{tamar2014policy, rajeswaran2016epopt,petersen2019deep,du2022discover} critic selection. Selection changes the trajectory distribution used for training, so the update below is not, in general, an unbiased policy-gradient estimator of Eq.~\eqref{eq:optim_problem}. In each iteration, we sample a batch of $B$ trajectories and compute their returns $R^{(i)}=\sum_{k=0}^{K-1}r(s_k^{(i)},\allowbreak a_k^{(i)})$. Let $\mathcal V$ index the valid complete trajectories. For the critic, we select the highest-return trajectories:
\begin{equation}
\label{eq:elite_set}
\mathcal{E}_{\epsilon}
=\operatorname{Top}_{m_\epsilon}\{R^{(i)}:i\in\mathcal V\},
\qquad m_\epsilon=\lceil\epsilon|\mathcal V|\rceil,
\end{equation}
where $\operatorname{Top}_{m}$ returns the indices of the $m$ largest returns, with ties resolved by sampling order. The actor instead uses all valid trajectories that achieve a dimension increase during the episode:
\[
\mathcal S_+=\{i\in\mathcal V:\max_{0\le k\le K}d(\Theta_k^{(i)})>d(\Theta_0)\}.
\]
Its policy-gradient loss is
\begin{equation}
\label{eq:actor_loss}
\mathcal{L}_{\mathrm{actor}}
=
-
\frac{1}{K|\mathcal S_+|}
\sum_{i\in\mathcal S_+}
\sum_{k=0}^{K-1}
\log \pi_\theta(a_k^{(i)}\mid s_k^{(i)})
\widehat A_k^{(i)},
\end{equation}
where $G_k^{(i)}=\sum_{j=k}^{K-1}r(s_j^{(i)},\allowbreak a_j^{(i)})$ is the return-to-go and $A_k^{(i)}=G_k^{(i)}-V_\phi(s_k^{(i)})$ is the advantage. We use $\widehat A_k^{(i)}=(A_k^{(i)}-\mu_A)/\max(\sigma_A,10^{-6})$, where $\mu_A$ and $\sigma_A$ are the mean and population standard deviation over all steps in $\mathcal S_+$. The normalized advantages are treated as constants in the actor update. If $\mathcal S_+=\emptyset$, we set $\mathcal L_{\mathrm{actor}}=0$ and retain the critic and entropy updates. The critic minimizes the mean squared error over $\mathcal E_\epsilon$:
\begin{equation}
\label{eq:critic_loss}
\mathcal{L}_{\mathrm{critic}}
=
\frac{1}{2K|\mathcal{E}_{\epsilon}|}
\sum_{i\in\mathcal{E}_{\epsilon}}
\sum_{k=0}^{K-1}
\left(
V_\phi(s_k^{(i)})-G_k^{(i)}
\right)^2 .
\end{equation}
To prevent premature convergence and encourage exploration, we add entropy regularization over all sampled trajectories, independently of the actor and critic selections. Let $\mathcal H_k^{(i)}$ be the sum of the decoder's conditional token entropies for action $a_k^{(i)}$. The final joint loss is:
\begin{equation}
\label{eq:total_loss}
\mathcal{L}
=
\mathcal{L}_{\mathrm{actor}}
+
\lambda_v\mathcal{L}_{\mathrm{critic}}
-
\lambda_H
\frac{1}{B}\sum_{i=1}^{B}\sum_{k=0}^{K-1}\mathcal H_k^{(i)}.
\end{equation}
If no valid complete trajectory is available, both actor and critic losses are zero, finite entropy terms can still contribute to the update.

\noindent\textbf{Basis Sparsification and Pruning.}
After policy optimization, we retain the best library found during search. We first eliminate symbolically verified library identities and retain the independent mapped generators. Let $\Theta^\ast$ be the reduced library, $S^\ast$ its retained column scales, and $Q$ the validated orthonormal basis of these generators in normalized coordinates, with $d^\ast$ columns. We sparsify this basis using an orthogonal rotation $P$:
\begin{equation}
\label{eq:LADMAP}
\min_{P,Z} \lVert Z \rVert_{1,1}, \quad \text{s.t.} \quad P^\top P=I, \quad Z=QP.
\end{equation}
We use LADMAP \citep{lin2011linearized} to obtain $P$ and set $Z=QP$. Reshaping the columns of $Z$ gives normalized coefficient matrices $\{\bar W_a\}_{a=1}^{d^\ast}$. We rank non-mandatory terms by their support in these coordinates:
\begin{equation}
\label{eq:normalized_support}
    s_j=\left(\sum_{a=1}^{d^\ast}\|\bar W_a[:,j]\|_2^2\right)^{1/2}.
\end{equation}
Low-support terms are tested for removal first. Removal requires unchanged independent generator dimension, acceptable conditioning, and a subspace distance below a prescribed tolerance. The reference subspace and retained column scales remain fixed.

After pruning, let $S_f$ contain the column scales of the final library $\Theta_f$. We recompute the null-space basis of $C(\Theta_f)S_f^{-1}$ using adaptive SVD and sparsify it again. Reshaping this sparse basis gives the final normalized coefficient matrices $\bar W_a$. We then recover the generators in the original library:
\begin{equation}
\label{eq:pruned_coefficient_recovery}
    \operatorname{vec}(W_a)=S_f^{-1}\operatorname{vec}(\bar W_a),\quad a=1,\ldots,d^\ast.
\end{equation}
The subspace criterion and Algorithm~\ref{alg:sparsification_pruning} are given in Appendix~\ref{app:algorithms}.

\subsection{Overall Algorithm of LieDiscover}
\label{sec:overall_algorithm}
Algorithm~\ref{alg:overall_framework} combines surrogate preparation, policy optimization, and basis extraction into the complete LieDiscover procedure. Episode indices are omitted within each rollout for brevity.

\begin{algorithm}[!t]
\caption{LieDiscover}
\label{alg:overall_framework}
\small
\begin{algorithmic}[1]
\REQUIRE Data $\mathcal{D}$, initial library $\Theta_0$, horizon $K$, batch size $B$, elite fraction $\epsilon$, iterations $N_{\mathrm{iter}}$
\ENSURE Compact library $\Theta_f$ and generators $\{W_j\}_{j=1}^{d^\ast}$
\STATE Fit and freeze $\widehat F$; fix $\mathcal I$ and cache $J_{\widehat F}$ by autodiff.
\STATE Initialize $\pi_\theta$, $V_\phi$, $R_{\mathrm{best}}\leftarrow-\infty$, $\Theta^\ast\leftarrow\Theta_0$.
\STATE $(\bar B_0,H_0,S_0)\leftarrow\mathrm{Normalize}(\Theta_0;\mathcal I)$.
\STATE $(d_0,\ldots)\leftarrow\mathrm{AdaptiveSVD}(\bar B_0,H_0,J_{\widehat F},\mathcal I)$; stop if $d_0=\bot$.
\FOR{$m=1,\dots,N_{\mathrm{iter}}$}
    \FOR{$i=1,\dots,B$}
        \STATE $s_0\leftarrow(\Theta_0,\emptyset,0,0)$.
        \FOR{$k=0,\dots,K-1$}
            \STATE Sample $a_k\sim\pi_\theta(\cdot\mid s_k)$; set $\widetilde\Theta_k\leftarrow\Theta_k\cup\{a_k\}$.
            \STATE $(\bar B_k,H_k,S_k)\leftarrow\mathrm{Normalize}(\widetilde\Theta_k;\mathcal I)$.
            \STATE $(d(\widetilde\Theta_k),\allowbreak\kappa(\widetilde\Theta_k),\ldots)\leftarrow\mathrm{AdaptiveSVD}(\bar B_k,H_k,J_{\widehat F},\mathcal I)$.
            \STATE Update $s_{k+1}$ and $r_k$ using Eqs.~\eqref{eq:library_transition}--\eqref{eq:reward equation}.
        \ENDFOR
        \STATE $R^{(i)}\leftarrow\sum_{k=0}^{K-1}r_k$.
        \STATE If $R^{(i)}>R_{\mathrm{best}}$, set $(R_{\mathrm{best}},\Theta^\ast)\leftarrow(R^{(i)},\Theta_{K})$.
    \ENDFOR
    \STATE Select successful episodes $\mathcal S_+$ for the actor and elite episodes $\mathcal E_\epsilon$ for the critic.
    \STATE Update $\theta,\phi$ using Eq.~\eqref{eq:total_loss}, with normalized advantages on $\mathcal S_+$.
\ENDFOR
\STATE $(\bar B^\ast,H^\ast,S^\ast)\leftarrow\mathrm{Normalize}(\Theta^\ast;\mathcal I)$.
\STATE $(d^\ast,\cdot,Q,\cdot)\leftarrow\mathrm{AdaptiveSVD}(\bar B^\ast,H^\ast,J_{\widehat F},\mathcal I)$.
\STATE $(\Theta_f,Q_{\mathrm{final}}^\ast)\leftarrow\mathrm{SparsifyAndPrune}(Q,\allowbreak\Theta^\ast,\allowbreak S^\ast,\allowbreak\Theta_0,\allowbreak J_{\widehat F},\allowbreak\mathcal I)$.
\RETURN $\Theta_f$ and $\{W_j\}_{j=1}^{d^\ast}\leftarrow\mathrm{Reshape}(Q_{\mathrm{final}}^\ast)$.
\end{algorithmic}
\end{algorithm}

\subsection{PDE Application}
\label{sec:PDE Application}
The compact library $\Theta_f$ and recovered generators define transformations for augmenting PDE solution data. We use these transformations in Lie Point Symmetry Data Augmentation (LPSDA)~\citep{brandstetter2022lie} to produce additional solution samples for training a Fourier Neural Operator (FNO)~\citep{li2020fourier}. This transfers the discovered symmetry information to the solver through its training data.

\begin{figure*}[pos=tb]
    \centering
    \includegraphics[width=0.325\linewidth]{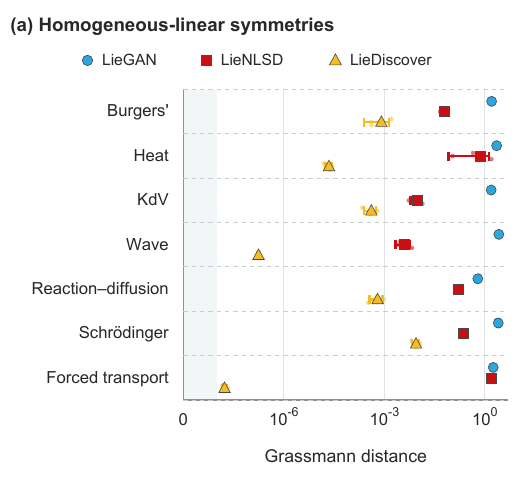}\hfill
    \includegraphics[width=0.325\linewidth]{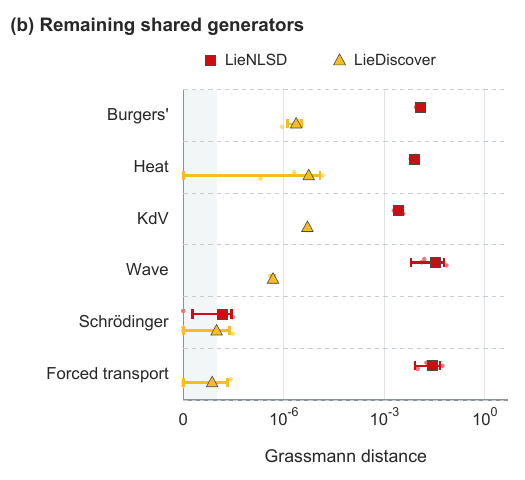}\hfill
    \includegraphics[width=0.325\linewidth]{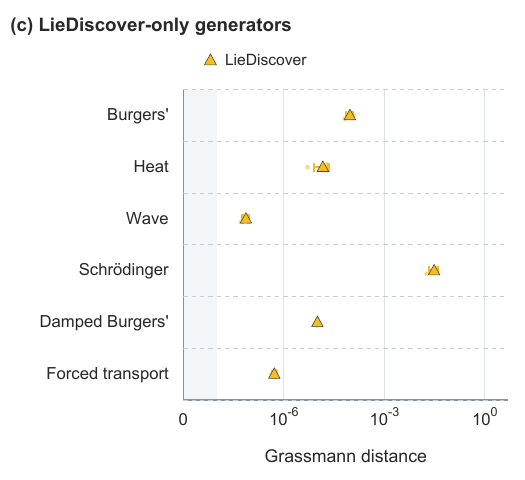}
    \caption{Grassmann distances for (a) homogeneous-linear symmetries (b) the remaining generators shared by LieNLSD and LieDiscover, and (c) generators recovered only by LieDiscover. Lower is better.}
    \label{fig:grassmann_comparison}
\end{figure*}

\section{Experiments}
Following LieNLSD, we evaluate LieDiscover on six dynamical systems: the Burgers', heat, KdV, wave, Schrödinger, and reaction-diffusion equations. To further test open-form symmetry discovery, we introduce two challenging benchmarks: the damped Burgers' \citep{montecinos2015analytic} and forced transport \citep{de1998multiple} equations. Further details on experimental settings, including computing infrastructure and hyperparameters, are provided in Appendix~\ref{app:settings}. Our experiments cover four aspects:
\begin{itemize}
    \item \textbf{Symmetry subspace recovery.} We compare discovered subspaces with analytic references using the Grassmann distance. This measures recovery accuracy.
    \item \textbf{Open-form symmetry discovery.} We examine the learned libraries and explicit generators for each equation. This tests the recovery of linear and nonlinear symmetries through library expansion.
    \item \textbf{Held-out validation.} We evaluate symmetry residuals on independent solution trajectories. This tests whether the recovered generators generalize beyond the fitting data.
    \item \textbf{PDE solving.} We use the discovered symmetries to augment FNO training data on Burgers', heat, and KdV, matching the three-equation scope of LieNLSD's downstream evaluation. This tests their effect on prediction accuracy.
\end{itemize}

\subsection{Symmetry Subspace Recovery}
\label{sec:Quantitative Metric}

We measure the discrepancy between discovered and analytic generator subspaces using the Grassmann distance. Let $Q_{\mathrm{est}},Q_{\mathrm{ref}}\in\mathbb R^{n\times d}$ be orthonormal bases for these subspaces in a common coefficient representation. Their distance is
\begin{equation}
    d_G(Q_{\mathrm{est}},Q_{\mathrm{ref}})=\sqrt{\sum_{i=1}^d\theta_i^2},\quad \theta_i=\arccos(\sigma_i),
\end{equation}
where $\sigma_i$ are the singular values of $Q_{\mathrm{est}}^\top Q_{\mathrm{ref}}$ and $\theta_i$ are the principal angles. Smaller values indicate closer agreement between the subspaces.

We organize the comparison into three groups. First, we compare LieGAN, LieNLSD, and LieDiscover on the same analytic generators with homogeneous-linear coefficients. LieGAN's learned subspace is compared directly; for LieNLSD and LieDiscover, the corresponding estimates are obtained by projecting the analytic directions into each learned generator space. Second, we compare LieNLSD and LieDiscover on their remaining shared analytic directions, including affine and nonlinear generators. Third, we evaluate the remaining analytic directions recovered only by LieDiscover. Each distance is measured against the analytic subspace for its group.

Fig.~\ref{fig:grassmann_comparison} summarizes all three groups, with detailed values in Appendix~\ref{app:grassmann_results}, including Table~\ref{tab:grassmann_exclusive} for the third group. Results are reported as means and standard deviations over three runs.

\subsection{Open-Form Symmetry Discovery}
\label{sec:Symmetry Discovery}
Table~\ref{tab:search_results} compares the final generator dimensions and representation library sizes of LieDiscover and LieNLSD. The results show that LieDiscover generally identifies more independent symmetry generators using a smaller representation library than LieNLSD, demonstrating the effectiveness of adaptive library construction in discovering richer symmetry spaces with fewer symbolic expressions.

Notably, for the damped Burgers' equation, LieDiscover introduces $e^{-t}$ and recovers the exponential generator $e^{-t}(\partial_x-\partial_u)$. For the forced transport equation, it discovers generators containing trigonometric functions of $x$ and $t-x$, including $\sin(x)$, $\cos(x)$, $\sin(t-x)$, and $\cos(t-x)$. Therefore, these cases demonstrate the recovery of non-polynomial symmetry structures through adaptive library construction.

 In this section, we present detailed discovery results for each PDE.

\begin{figure*}[pos=tb]
    \centering
    \includegraphics[width=0.95\linewidth]{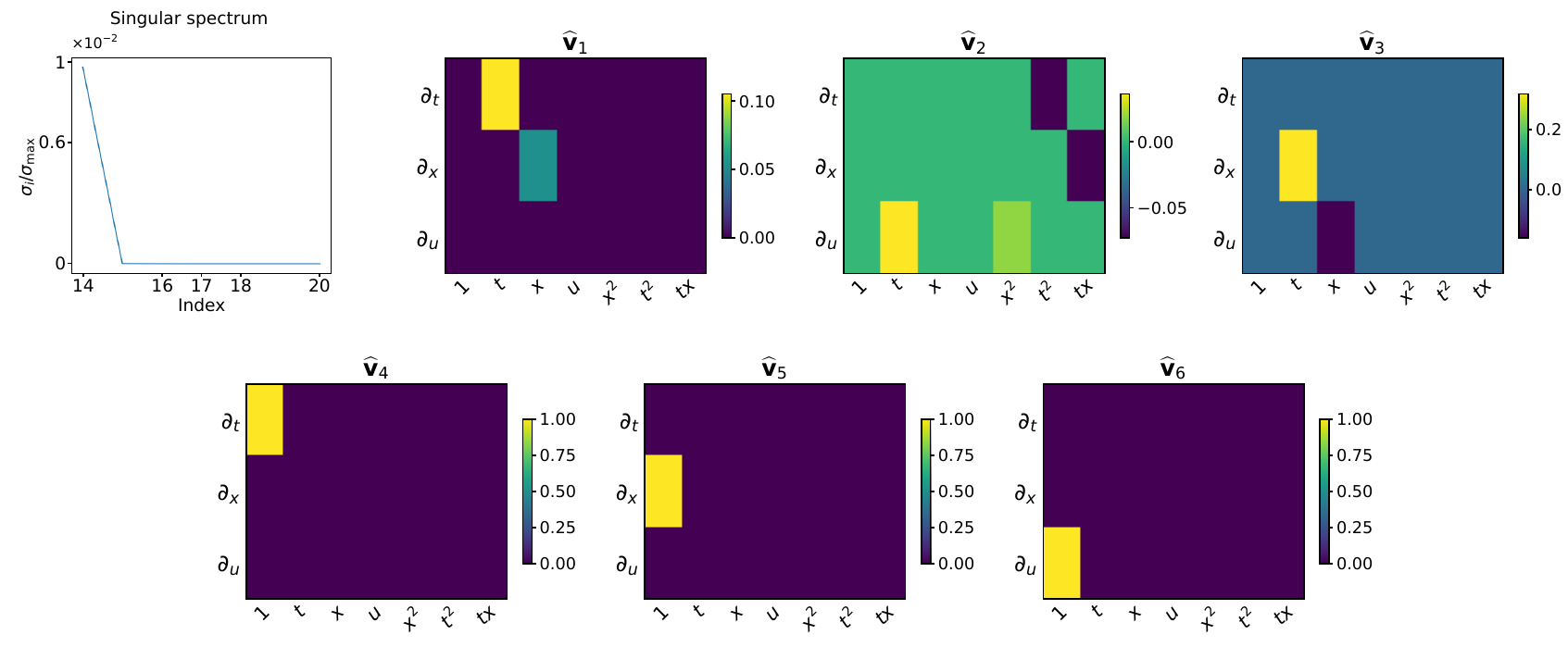}
    \caption{Symmetry discovery for the potential Burgers' equation (seed $0$). The panels show the normalized singular spectrum and 6 sparse generator coefficient matrices.}
    \label{fig:burgers'}
\end{figure*}

\begin{table}[pos=t]
\centering
\footnotesize
\setlength{\tabcolsep}{2.5pt}
\caption{Final generator-space dimensions $d_f$ and library sizes $|\Theta_f|$. Entries follow seed order $0,1,2$, except Burgers\textquoteright{}, which uses $0,2,3$. LieNLSD uses fixed candidate libraries.}
\label{tab:search_results}
\begin{tabular}{llcc}
\toprule
Dataset & Method & $d_f$ & $|\Theta_f|$ \\
\midrule
Burgers' & LieDiscover & $6,6,6$ & $7,7,7$ \\
 & LieNLSD & $5,5,5$ & $10,10,10$ \\
\addlinespace[2pt]
Heat & LieDiscover & $9,9,9$ & $8,8,8$ \\
 & LieNLSD & $5,7,7$ & $10,10,10$ \\
\addlinespace[2pt]
KdV & LieDiscover & $4,4,4$ & $4,4,4$ \\
 & LieNLSD & $4,4,4$ & $10,10,10$ \\
\addlinespace[2pt]
Damped Burgers' & LieDiscover & $3,3,3$ & $5,5,5$ \\
 & LieNLSD & $2,2,2$ & $10,10,10$ \\
\addlinespace[2pt]
Forced transport & LieDiscover & $16,15,15$ & $11,10,11$ \\
 & LieNLSD & $4,4,4$ & $10,10,10$ \\
\addlinespace[2pt]
Wave & LieDiscover & $22,22,22$ & $17,19,18$ \\
 & LieNLSD & $20,20,20$ & $15,15,15$ \\
\addlinespace[2pt]
Schr\"odinger & LieDiscover & $8,8,8$ & $15,18,18$ \\
 & LieNLSD & $6,6,6$ & $21,21,21$ \\
\addlinespace[2pt]
Reaction-diffusion & LieDiscover & $5,5,5$ & $6,6,10$ \\
 & LieNLSD & $5,5,5$ & $21,21,21$ \\
\bottomrule
\end{tabular}
\end{table}

\begin{table*}[pos=tb]
\centering
\small
\setlength{\tabcolsep}{2.5pt}
\caption{Analytic reference bases for the representative seed-$0$ subspaces. The corresponding figures show the recovered numerical coefficients. For forced transport, $r=t-x$, $z=u+\cos x$, and $H=\{1,r,\sin r,\cos r,\sin^2r,\sin^4r\}$.}
\label{tab:selected_generators}
\begin{tabular}{
c|
>{\raggedright\arraybackslash}p{0.175\textwidth}|
>{\raggedright\arraybackslash}p{0.105\textwidth}|
>{\raggedright\arraybackslash}p{0.175\textwidth}|
>{\raggedright\arraybackslash}p{0.40\textwidth}
}
\toprule
Dataset
& Heat 
& Damped 
& Forced
& Wave  \\
\midrule
Generators
&
\(\begin{aligned}[t]
\mathbf v_1 &= (2t+x^2)\partial_u,\\
\mathbf v_2 &= 2t\partial_t+x\partial_x,\\
\mathbf v_3 &= 2t\partial_x-xu\partial_u,\\
\mathbf v_4 &= x\partial_u,\
\mathbf v_5 = u\partial_u,\\
\mathbf v_6 &= \partial_u,\
\mathbf v_7 = \partial_x,\\
\mathbf v_8 &= \partial_t,\\
\mathbf {v_9} &= (6tx+x^3)\partial_u.
\end{aligned}\)
&
\(\begin{aligned}[t]
\mathbf v_1 &= \partial_t,\\
\mathbf v_2 &= \partial_x,\\
\mathbf v_3 &= e^{-t}\partial_x\\
&\ -e^{-t}\partial_u.
\end{aligned}\)
&
\(\begin{aligned}[t]
&h\partial_t,\ h\partial_u\quad(h\in H),\\
&z\partial_t,\ z\partial_u,\\
&\partial_x+\sin x\,\partial_u,\\
&\sin x\,(\partial_t+\partial_x)\\
&\qquad+\sin^2x\,\partial_u.
\end{aligned}\)
&
\(\begin{aligned}[t]
\mathbf v_1 &= 2tx\partial_t+(t^2+x^2-y^2)\partial_x+2xy\partial_y-xu\partial_u,\\
\mathbf v_2 &= 2ty\partial_t+2xy\partial_x+(t^2-x^2+y^2)\partial_y-yu\partial_u,\\
\mathbf v_3 &= x\partial_t+t\partial_x,\ 
\mathbf v_4 = y\partial_t+t\partial_y,\\
\mathbf v_5 &= t\partial_t+x\partial_x+y\partial_y,\ 
\mathbf v_6 = -y\partial_x+x\partial_y,\\
\mathbf v_7 &= u\partial_u,\
\mathbf v_8 = \partial_t,\
\mathbf v_9 = \partial_x,\
\mathbf v_{10} = \partial_y,\\
\mathbf v_{11} &= \partial_u,\
\mathbf v_{12}=t\partial_u,\
\mathbf v_{13}=x\partial_u,\
\mathbf v_{14}=y\partial_u,\\
\mathbf v_{15} &= tx\partial_u,\
\mathbf v_{16}=ty\partial_u,\
\mathbf v_{17}=xy\partial_u,\
\mathbf v_{18}=txy\partial_u,\\
\mathbf v_{19} &= (t^2+x^2)\partial_u,\
\mathbf v_{20}=(t^2+y^2)\partial_u,\\
\mathbf v_{21} &= (3t^2x+x^3)\partial_u,\
\mathbf v_{22}=(t^2x+xy^2)\partial_u.
\end{aligned}\)
\\
\bottomrule
\end{tabular}

\end{table*}

\begin{figure*}[pos=tb]
    \centering
    \includegraphics[width=\linewidth]{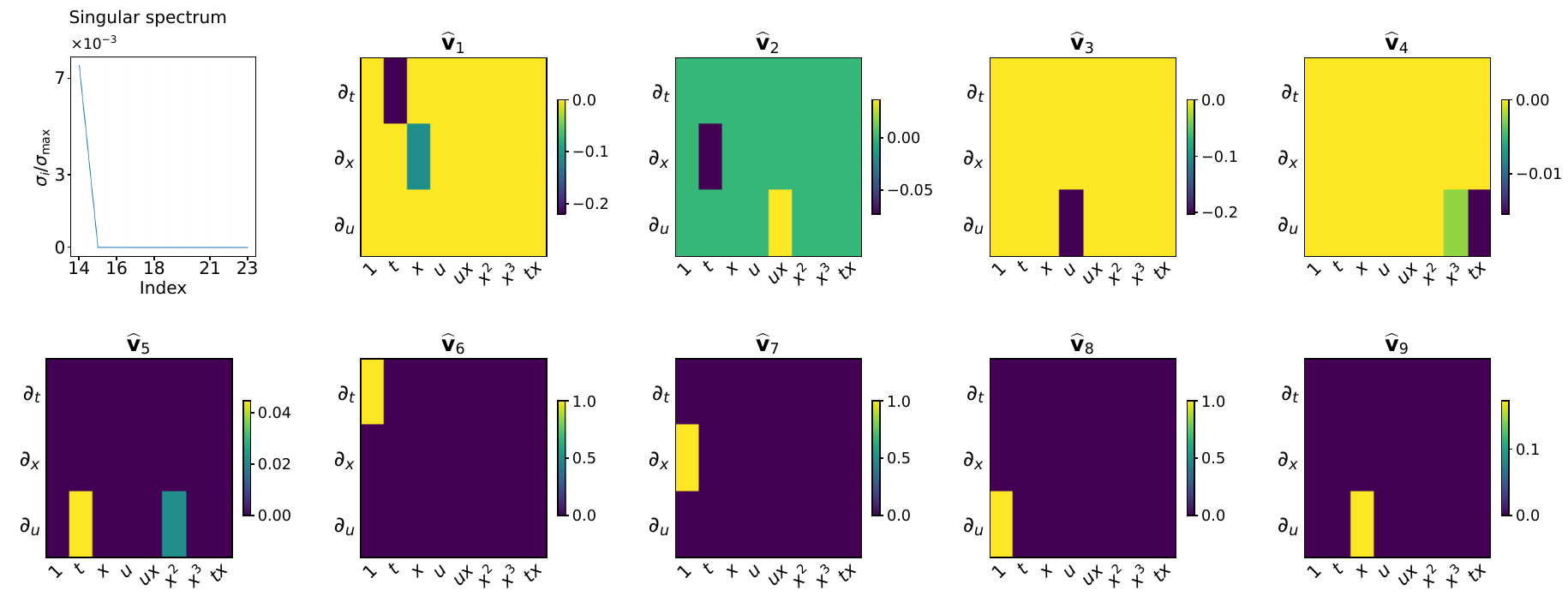}
    \caption{Symmetry discovery for the heat equation (seed $0$). The panels show the normalized singular spectrum and 9 sparse generator coefficient matrices.}
    \label{fig:heat}
\end{figure*}

\subsubsection{Burgers' Equation}

For the potential Burgers' equation \citep{hopf1950partial}, we consider the evolution equation $u_t=u_{xx}+u_x^2$. This is the potential Burgers' benchmark used by LieNLSD \citep{hu2025explicit}.

We initialize the candidate library as
\(\Theta_0=\{1,t,x,u\}\), and use \(K=6\). In seed $0$, the selected library is $\Theta^\ast=\{1,\allowbreak t,\allowbreak x,\allowbreak u,\allowbreak x^2,\allowbreak t^2,\allowbreak tx,\allowbreak tu,\allowbreak tux,\allowbreak u^2\}$. Pruning gives $\Theta_f=\{1,\allowbreak t,\allowbreak x,\allowbreak u,\allowbreak x^2,\allowbreak t^2,\allowbreak tx\}$. All three reported runs increase the numerical dimension from $5$ to $6$ and retain these seven terms.

A reference basis for this six-dimensional subspace is
\begin{equation}
    \begin{aligned}
        \mathbf v_1 &= 4t^2\partial_t+4tx\partial_x-(2t+x^2)\partial_u,\\
        \mathbf v_2 &= 2t\partial_t+x\partial_x,\\
        \mathbf v_3 &= 2t\partial_x-x\partial_u,\\
        \mathbf v_4 &= \partial_u,\quad
        \mathbf v_5 = \partial_t,\quad
        \mathbf v_6 = \partial_x .
    \end{aligned}
\end{equation}
\(\mathbf v_5\) and \(\mathbf v_6\) correspond to time and space translations, \(\mathbf v_4\) is a vertical translation in \(u\), \(\mathbf v_2\) is a scaling symmetry, while \(\mathbf v_1\) and \(\mathbf v_3\) correspond to projective and Galilean-type transformations. LieDiscover recovers this six-dimensional polynomial subspace from the pruned library.

For Burgers' equation, Fig.~\ref{fig:burgers'} visualizes the singular values and the sparsified infinitesimal generators discovered by LieDiscover.

\begin{figure*}[pos=tb]
    \centering
    \includegraphics[width=\linewidth]{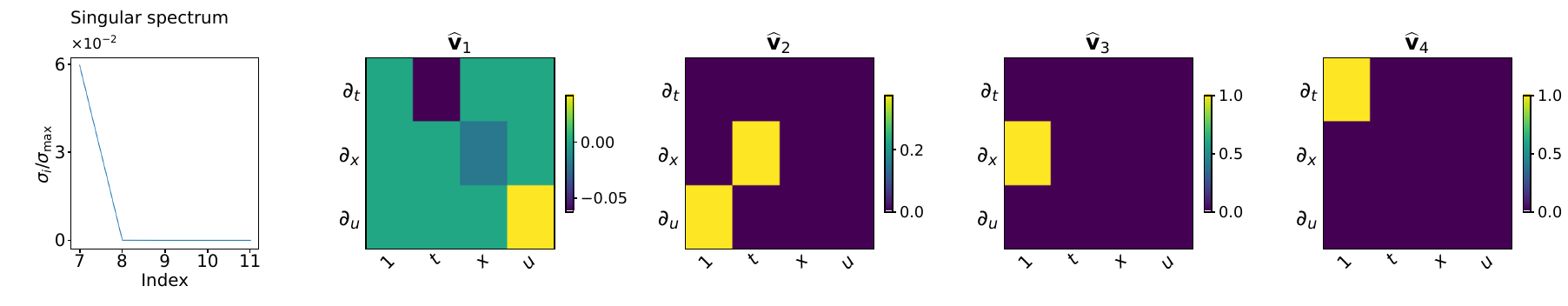}
    \caption{Symmetry discovery for the KdV equation (seed $0$). The panels show the normalized singular spectrum and 4 sparse generator coefficient matrices.}
    \label{fig:KdV}
\end{figure*}

\subsubsection{Heat Equation}

For the heat equation \citep{baron1822theorie}, we consider $u_t=u_{xx}$. We initialize \(\Theta_0=\{1,t,x,u\}\), and use \(K=6\). In seed $0$, the selected library is $\Theta^\ast=\{1,\allowbreak t,\allowbreak x,\allowbreak u,\allowbreak ux,\allowbreak x^2,\allowbreak x^3,\allowbreak tx,\allowbreak u^2\}$.
Pruning removes $u^2$ and gives $\Theta_f=\{1,\allowbreak t,\allowbreak x,\allowbreak u,\allowbreak ux,\allowbreak x^2,\allowbreak x^3,\allowbreak tx\}$. All three runs retain this set of eight terms.

The numerical null-space dimension increases from $6$ to $9$ in all three runs. Table~\ref{tab:selected_generators} gives the corresponding reference basis, comprising five geometric symmetries and four polynomial representatives of the linear superposition symmetry. More generally, this infinite-dimensional family has the form $\mathbf v_\psi=\psi(t,\allowbreak x)\partial_u$, where $\psi_t=\psi_{xx}$. The associated group action maps a solution $u=f(t,\allowbreak x)$ to $f(t,\allowbreak x)+\varepsilon_g\psi(t,\allowbreak x)$, where $\varepsilon_g$ denotes the group parameter. Thus, unlike nonlinear equations with a finite-dimensional symmetry algebra, the heat equation naturally yields additional solution generators. LieDiscover captures the geometric symmetries and finite polynomial representatives of this infinite-dimensional linear superposition symmetry.

In particular, LieDiscover identifies the cubic solution mode $(6tx+x^3)\partial_u$. Fig.~\ref{fig:heat} visualizes the singular values and the sparsified infinitesimal generators discovered by LieDiscover.

\begin{figure*}[pos=tb]
    \centering
    \includegraphics[width=\linewidth]{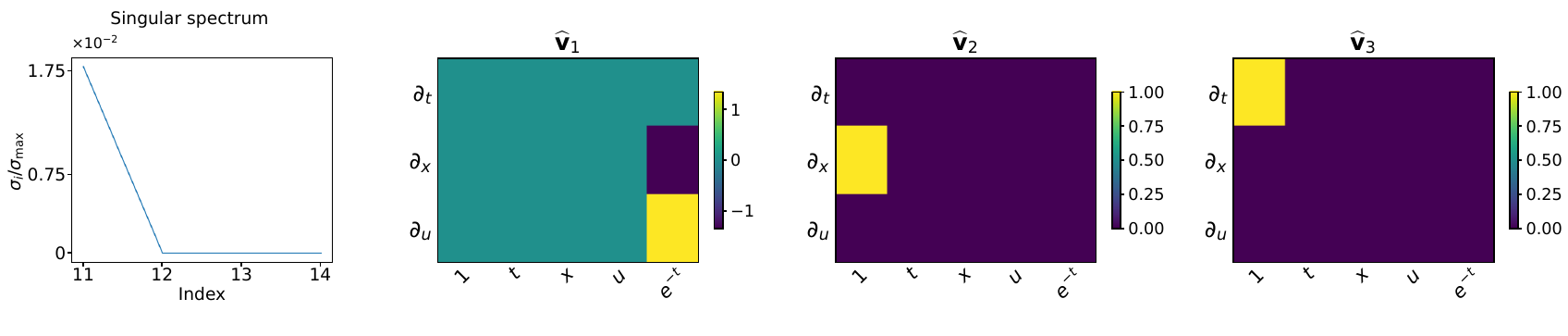}
    \caption{Symmetry discovery for the damped Burgers' equation (seed $0$). The panels show the normalized singular spectrum and 3 sparse generator coefficient matrices.}
    \label{fig:damped_burger}
\end{figure*}

\subsubsection{KdV Equation}

The KdV equation \citep{korteweg1895change} is a nonlinear dispersive PDE. We consider the standard form $u_t+u_{xxx}+uu_x=0$.

We initialize $\Theta_0=\{1,\allowbreak t,\allowbreak x,\allowbreak u\}$, and use $K=6$. In seed $0$, the selected library is $\Theta^\ast=\{1,\allowbreak t,\allowbreak x,\allowbreak u,\allowbreak ux,\allowbreak tx,\allowbreak t^2,\allowbreak u^2,\allowbreak t^4,\allowbreak x^2\}$. The numerical dimension remains $d=4$ in all three runs. Pruning returns each library to $\Theta_f=\{1,\allowbreak t,\allowbreak x,\allowbreak u\}$.

This library is sufficient to recover the four-dimensional symmetry algebra:
\begin{equation}
    \begin{aligned}
        \mathbf v_1 &= -3t\partial_t-x\partial_x+2u\partial_u,\\
        \mathbf v_2 &= t\partial_x+\partial_u,\\
        \mathbf v_3 &= \partial_t,\quad
        \mathbf v_4=\partial_x.
    \end{aligned}
\end{equation}
\(\mathbf v_3\) and \(\mathbf v_4\) are time and space translations, \(\mathbf v_1\) is a scaling symmetry, and \(\mathbf v_2\) is the Galilean-type symmetry of the KdV equation. The result shows that LieDiscover does not unnecessarily enlarge the library when the initial terms are already sufficient.

Fig.~\ref{fig:KdV} visualizes the singular values and the sparsified infinitesimal generators discovered by LieDiscover. The four nearly-zero singular values indicate that LieDiscover correctly identifies a four-dimensional Lie algebra subspace.

\begin{figure*}[pos=tb]
    \centering
    \includegraphics[width=\linewidth,height=0.50\textheight,keepaspectratio]{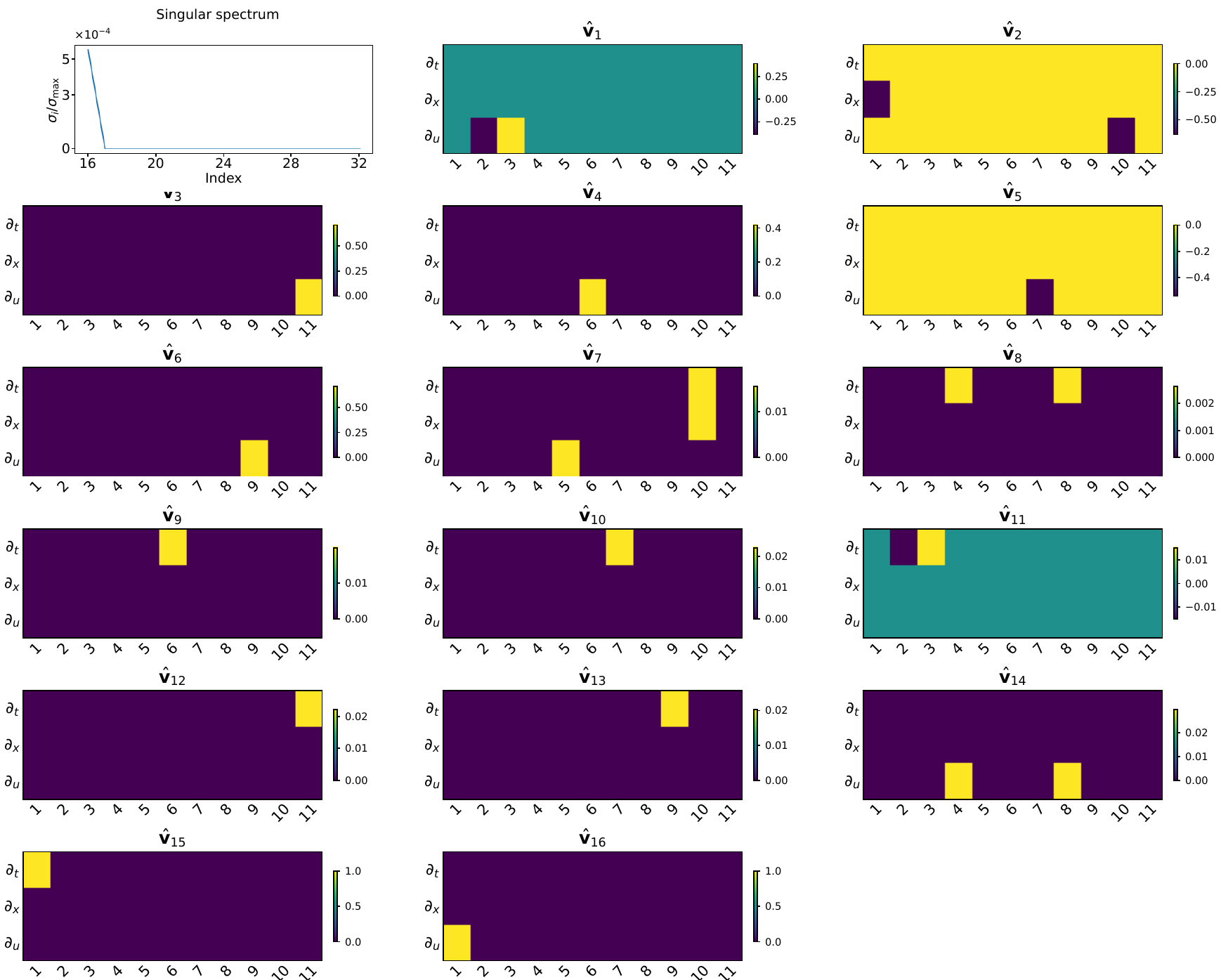}
    \caption{Symmetry discovery for the forced transport equation (seed $0$). The panels show the normalized singular spectrum and 16 sparse generator coefficient matrices. Column indices $1$--$11$ correspond, in order, to $1$, $t$, $x$, $u$, $\sin^2(x)$, $\sin^4(t-x)$, $\sin^2(t-x)$, $\cos(x)$, $-\sin(t-x)$, $\sin(x)$, and $\cos(t-x)$.}
    \label{fig:forced transport}
\end{figure*}

\subsubsection{Damped Burgers' Equation}

For the damped Burgers' equation \citep{montecinos2015analytic}, we consider
$u_t=u_{xx}-uu_x-\mu u$. The damping term changes the Galilean symmetry of the standard Burgers equation and introduces exponential factors in the infinitesimal generators. We generate a new dataset because this task is not included in the original benchmarks in LieNLSD. In our experiments, we set \(\mu=1\). The data are generated on a periodic one-dimensional spatial domain from random Fourier initial conditions. We evolve the PDE forward in time using a finite-difference spatial discretization with periodic boundary conditions and an explicit time integrator. The resulting trajectories are used to compute the required jet variables, including \(u,u_x,u_{xx}\), and \(u_t\).

We initialize $\Theta_0=\{1,\allowbreak t,\allowbreak x,\allowbreak u\}$, and use \(K=6\). In seed $0$, the search yields $\Theta^\ast=\{1,\allowbreak t,\allowbreak x,\allowbreak u,\allowbreak e^{-t},\allowbreak tu,\allowbreak e^t,\allowbreak t^2x,\allowbreak tx,\allowbreak x^2\}$. All three runs reduce to $\Theta_f=\{1,\allowbreak t,\allowbreak x,\allowbreak u,\allowbreak e^{-t}\}$ after pruning.

The numerical dimension increases from $2$ to $3$ in all three runs. The reference basis is given in Table~\ref{tab:selected_generators}.
The first two generators correspond to time and space translations. The third generator is a damping-modified pseudo-Galilean symmetry: the spatial shift and the compensation in \(u\) decay exponentially in time. This example demonstrates that LieDiscover can extend the library beyond polynomials and recover non-polynomial symmetry structures.

Fig.~\ref{fig:damped_burger} visualizes the singular values and the sparsified infinitesimal generators discovered by LieDiscover.

\begin{figure*}[pos=tb]
    \centering
    \includegraphics[width=0.8\linewidth,height=\textheight,keepaspectratio]{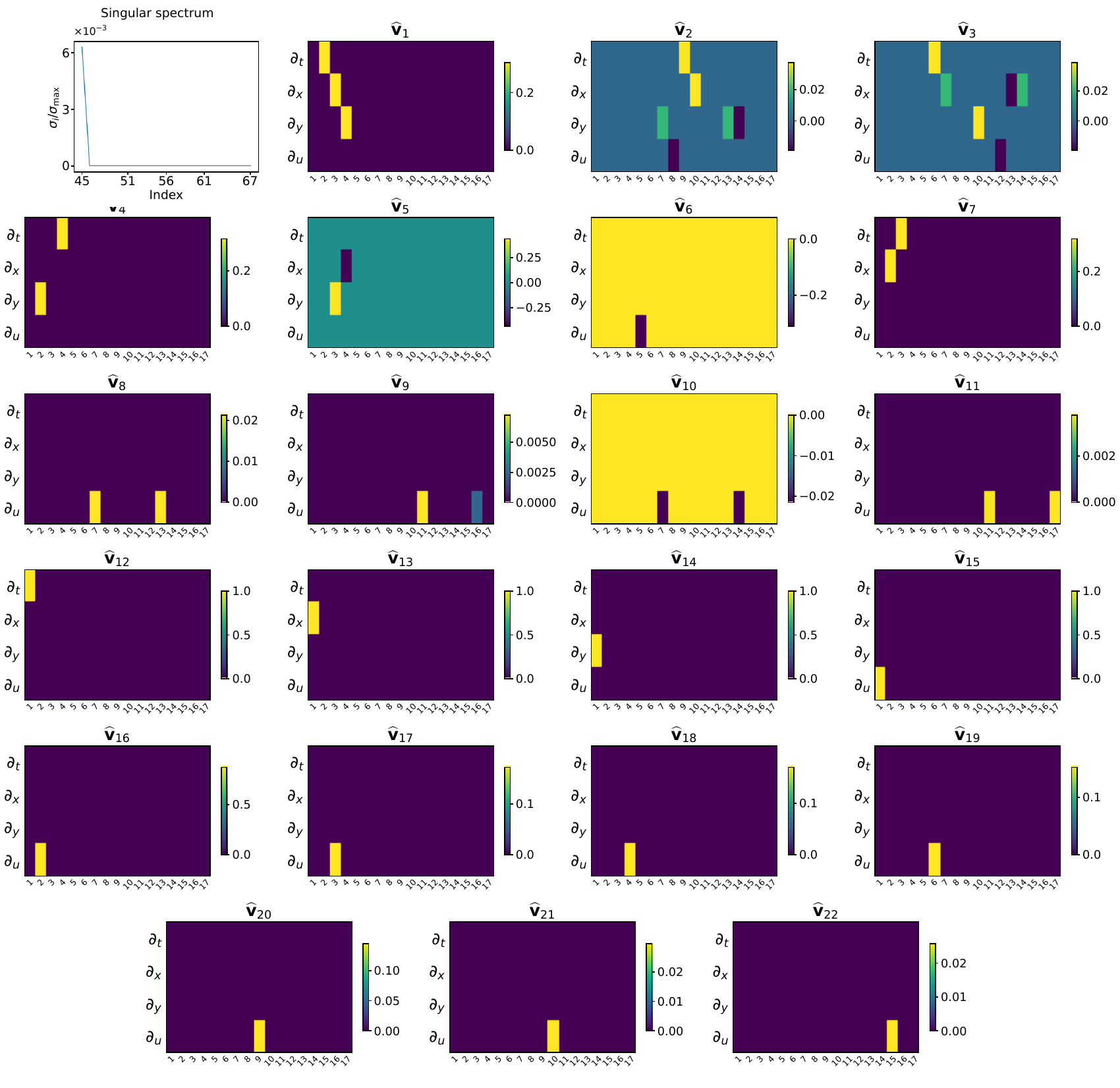}
    \caption{Symmetry discovery for the wave equation (seed $0$). The panels show the normalized singular spectrum and 22 sparse generator coefficient matrices. Column indices $1$--$17$ correspond, in order, to $1$, $t$, $x$, $y$, $u$, $tx$, $t^2$, $uy$, $ty$, $xy$, $t^2x$, $ux$, $y^2$, $x^2$, $txy$, $x^3$, and $xy^2$.}
    \label{fig:wave}
\end{figure*}

\subsubsection{Forced Transport Equation}

For the forced transport equation \citep{de1998multiple}, we consider \(u_t+c u_x=\sin(x)\) with \(c=1\) on a periodic one-dimensional domain. We sample random Fourier initial conditions and use the exact characteristic solution to compute \(u(t,x)\), \(u_x(t,x)\), \(u_{xx}(t,x)\), and \(u_t(t,x)\), which form the training data.

We initialize $\Theta_0=\{1,\allowbreak t,\allowbreak x,\allowbreak u\}$. Prior constraints restrict trigonometric arguments and exclude nested expressions such as $\sin(\cos(\cdot))$~\citep{petersen2019deep}.

The initial numerical dimension is $4$. The final dimensions are $16$, $15$, and $15$ for seeds $0$, $1$, and $2$. For seed $0$, write $r=t-x$. Pruning removes $-x\sin r$ from the selected library and gives
\[
\begin{aligned}
\Theta_f=\{&1,t,x,u,\sin^2x,\sin^4r,\sin^2r,\\
          &\cos x,-\sin r,\sin x,\cos r\}.
\end{aligned}
\]

Table~\ref{tab:selected_generators} gives a reference basis for this 16-dimensional subspace. It includes six pairs $h(r)\partial_t$ and $h(r)\partial_u$, two directions involving $z=u+\cos x$, and two transformations that compensate for the forcing. These directions are linearly independent over constant coefficients. 

The equation has infinite-dimensional symmetry families. For example, $\psi(r)\partial_u$ is a symmetry for any smooth $\psi$. The finite dimension therefore depends on which functions the search includes. The difference between $16$ and $15$ across runs reflects different finite subspaces and does not by itself indicate a false generator.

Fig.~\ref{fig:forced transport} visualizes the singular values and the sparsified infinitesimal generators discovered by LieDiscover.

\begin{figure*}[pos=tb]
    \centering
    \includegraphics[width=0.8\linewidth]{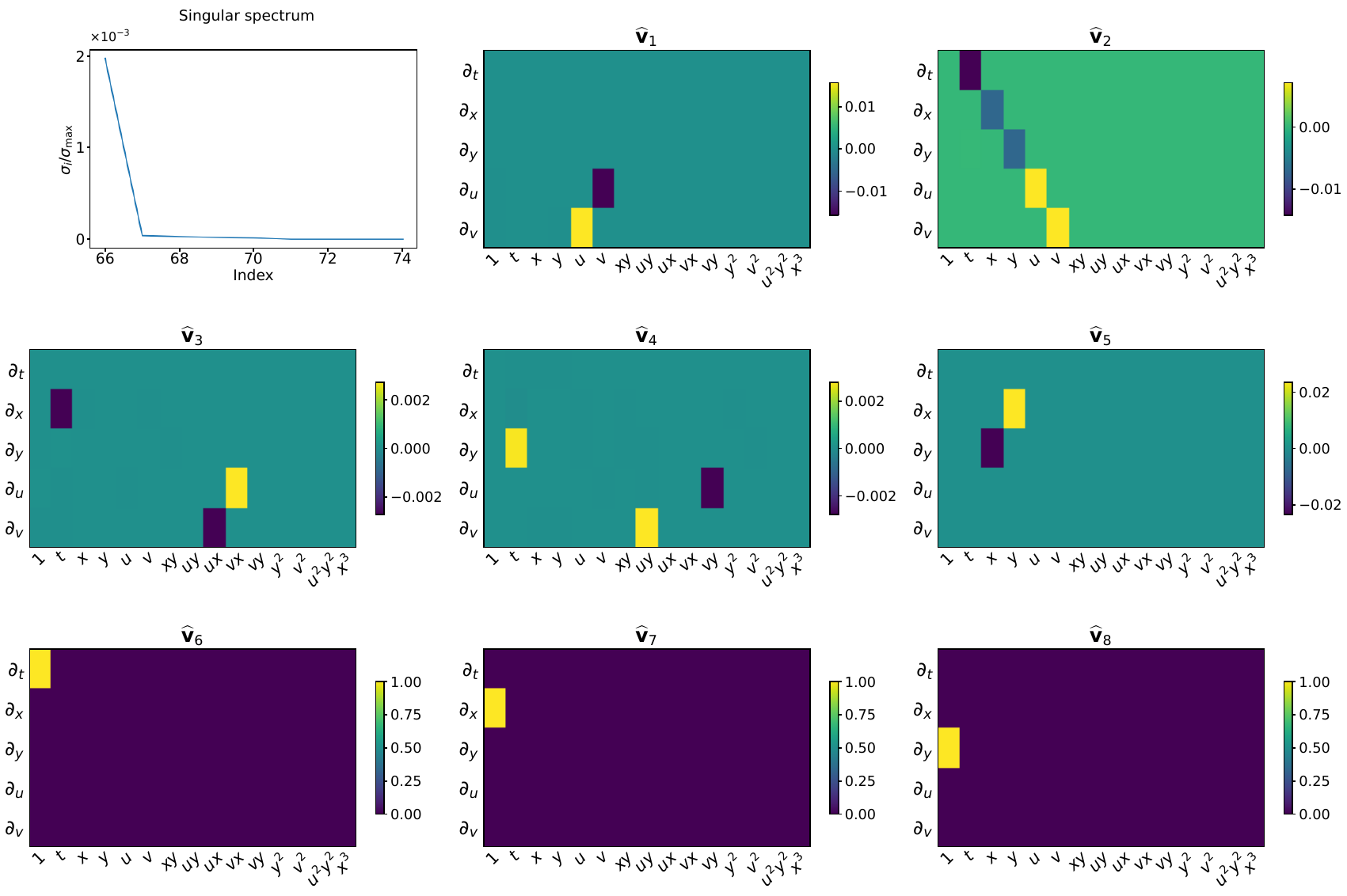}
    \caption{Symmetry discovery for the Schr\"odinger equation (seed $0$). The panels show the normalized singular spectrum and 8 sparse generator coefficient matrices.}
    \label{fig:Schrodinger}
\end{figure*}

\subsubsection{Wave Equation}

For the two-dimensional wave equation \citep{strauss2007partial}, we consider $u_{tt}=u_{xx}+u_{yy}$.

We initialize $\Theta_0=\{1,\allowbreak t,\allowbreak x,\allowbreak y,\allowbreak u\}$, and use \(K=15\). The numerical dimension increases from $12$ to $22$ in all three runs. The final libraries contain $17$, $19$, and $18$ terms. In seed $0$, pruning removes $tx^3$ and $t^2y$ from the 19-term selected library, leaving
\[
\begin{aligned}
\Theta_f=\{&1,t,x,y,u,tx,t^2,uy,ty,xy,t^2x,\\
           &ux,y^2,x^2,txy,x^3,xy^2\}.
\end{aligned}
\]

The wave equation contains translations, rotations, scalings, conformal-type transformations, and a large family of solution-superposition symmetries. LieDiscover recovers a 22-dimensional finite subspace in each run. Table~\ref{tab:selected_generators} gives a reference basis for the seed-$0$ result.
In addition, LieDiscover recovers translations \(\partial_t,\partial_x,\partial_y\), scaling \(u\partial_u\), and multiple polynomial solution modes such as
\(\partial_u,t\partial_u,x\partial_u,y\partial_u,tx\partial_u,ty\partial_u,xy\partial_u\), and higher-order modes such as $(3t^2x+x^3)\partial_u$ and $(t^2x+xy^2)\partial_u$. These modes are finite-dimensional representatives of the family $\mathbf v_\psi=\psi(t,\allowbreak x,\allowbreak y)\partial_u$, where $\psi_{tt}=\psi_{xx}+\psi_{yy}$. Other seeds select different polynomial representatives of this family.

Fig.~\ref{fig:wave} visualizes the singular values and the sparsified infinitesimal generators discovered by LieDiscover.

\begin{figure*}[pos=tb]
    \centering
    \includegraphics[width=0.8\linewidth]{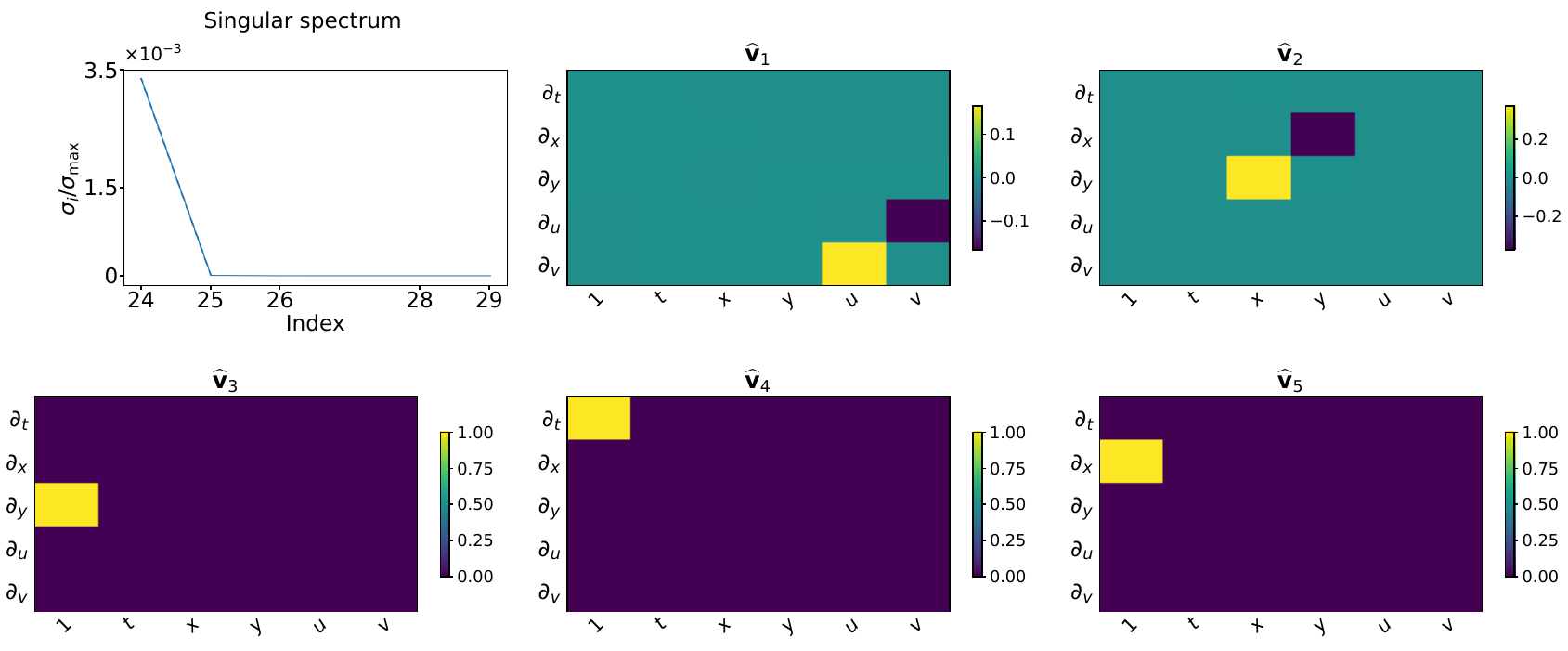}
    \caption{Symmetry discovery for the reaction-diffusion system (seed $0$). The panels show the normalized singular spectrum and five sparse generator coefficient matrices.}
    \label{fig:RD}
\end{figure*}

\subsubsection{Schr\"odinger Equation}
\label{sec:schrodinger_results}

The nonlinear Schr\"odinger equation~\citep{sulem1999nonlinear} is a fundamental dispersive PDE that models wave propagation in nonlinear media, including nonlinear optics and Bose--Einstein condensates. Here, the complex-valued wave field is written in terms of its real and imaginary components, \(u\) and \(v\), which couple through nonlinear self-interaction and spatial dispersion.
\[
\begin{aligned}
    u_t &= -\frac{1}{2}(v_{xx}+v_{yy}) + v(u^2+v^2),\\
    v_t &= \frac{1}{2}(u_{xx}+u_{yy}) - u(u^2+v^2),
\end{aligned}
\]
We initialize $\Theta_0=\{1,\allowbreak t,\allowbreak x,\allowbreak y,\allowbreak u,\allowbreak v\}$, and use $K=24$. The numerical dimension increases from $6$ to $8$ in all three runs. In seed $0$, pruning reduces the selected library from 18 to 15 terms:
\[
\begin{aligned}
\Theta_f=\{&1,t,x,y,u,v,xy,uy,ux,vx,vy,\\
          &y^2,v^2,u^2y^2,x^3\}.
\end{aligned}
\]
The numerical subspace approximates the span of
\begin{equation}
    \begin{aligned}
    \mathbf v_1 &= -2t\partial_t-x\partial_x-y\partial_y
    +u\partial_u+v\partial_v,\\
    \mathbf v_2 &= -v\partial_u+u\partial_v,
 \end{aligned}
\end{equation}
\begin{equation*}
    \begin{aligned}
    \mathbf v_3 &= -y\partial_x+x\partial_y,\\
    \mathbf v_4 &= \partial_t,\quad
    \mathbf v_5=\partial_x,\quad
    \mathbf v_6=\partial_y,\\
    \mathbf v_7 &= t\partial_x-xv\partial_u+xu\partial_v,\\
    \mathbf v_8 &= t\partial_y-yv\partial_u+yu\partial_v.
    \end{aligned}
\end{equation*}
The first three generators describe scaling, phase rotation, and spatial rotation. The next three are translations. The last two are Galilean boosts and require the mixed terms $ux,vx,uy,vy$. Fig.~\ref{fig:Schrodinger} visualizes the singular values and the sparsified infinitesimal generators discovered by LieDiscover.

LieDiscover recovers eight of the nine reference symmetry directions. The missing direction is the projective generator
\[
\begin{aligned}
\mathbf v_{\mathrm{proj}}^{\mathrm{Sch}}
={}&t^2\partial_t+tx\partial_x+ty\partial_y\\
&-t(u\partial_u+v\partial_v)+\frac{x^2+y^2}{2}(u\partial_v-v\partial_u).
\end{aligned}
\]
Thus, the three-dimensional homogeneous-linear score in Table~\ref{tab:grassmann_distance} should not be interpreted as recovery of the complete nine-dimensional algebra.

In a seed-$0$ diagnostic, we fixed a 19-term library spanning all nine directions and expanded candidate selection to eight spectral gaps. With 100 points in each of three fitting subsets, the nine-dimensional candidate had a validation residual of $2.34\times10^{-4}$, below the $10^{-3}$ tolerance. However, its coefficient-space and function-space stability errors were $0.0949$ and $0.1101$, respectively, both exceeding the $0.01$ threshold. Thus, even with all required terms present, the nine-dimensional candidate was rejected because its estimated subspace varied across data subsets. We discussed this issue in detail in Section \ref{sec:discussion}.

\begin{figure*}[pos=tb]
\centering
\includegraphics[width=\linewidth]{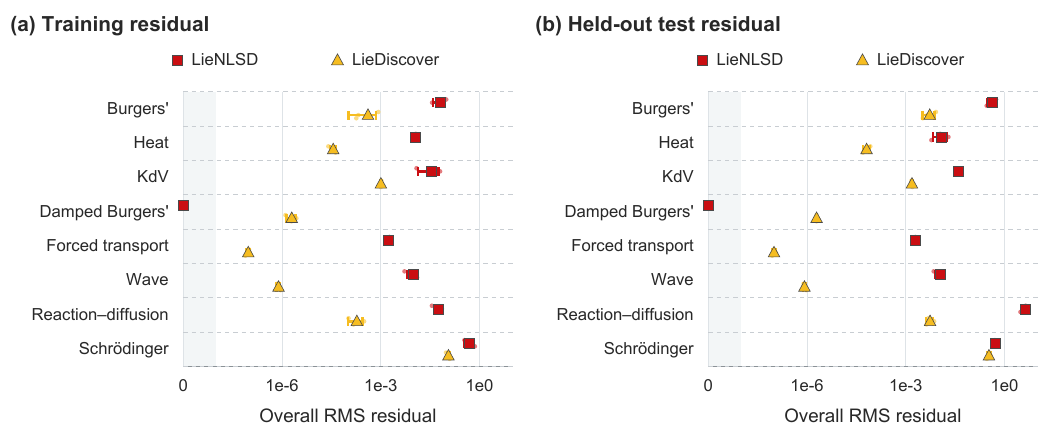}
\caption{Training and held-out symmetry residuals for LieNLSD and LieDiscover. Large markers and error bars show means and population standard deviations over three runs; small markers show individual runs. Lower is better. Both panels use the same axis, linear from zero to $10^{-8}$ (shaded) and logarithmic above $10^{-8}$.}
\label{fig:heldout_comparison}
\end{figure*}

\subsubsection{Reaction-diffusion Equation}

We consider the following two-dimensional reaction--diffusion system of $\lambda$--$\omega$ type~\citep{kopell1973plane}
\[
\begin{aligned}
    u_t &= (1-u^2-v^2)u+(u^2+v^2)v+0.1(u_{xx}+u_{yy}),\\
    v_t &= -(u^2+v^2)u+(1-u^2-v^2)v+0.1(v_{xx}+v_{yy}).
\end{aligned}
\]
The reaction--diffusion system models the spatiotemporal evolution of two interacting components under nonlinear local reactions and diffusion. In this system, the terms involving \(u^2+v^2\) describe nonlinear coupling between the two fields, while the Laplacian terms account for spatial diffusion.

We initialize $\Theta_0=\{1,\allowbreak t,\allowbreak x,\allowbreak y,\allowbreak u,\allowbreak v\}$, and use $K=15$. The numerical dimension remains $5$ in all three runs. In seed $0$, the search selects an 18-term library. Pruning returns it to $\Theta_f=\{1,\allowbreak t,\allowbreak x,\allowbreak y,\allowbreak u,\allowbreak v\}$. Seed $1$ gives the same compact library, while seed $2$ retains ten terms without increasing the dimension.

A reference basis is
\begin{equation}
    \begin{aligned}
        \mathbf v_1 &= -v\partial_u+u\partial_v,\\
        \mathbf v_2 &= -y\partial_x+x\partial_y,\\
        \mathbf v_3 &= \partial_t,\quad
        \mathbf v_4=\partial_x,\quad
        \mathbf v_5=\partial_y .
    \end{aligned}
\end{equation}
\(\mathbf v_1\) is the internal phase rotation in the \((u,v)\) plane, \(\mathbf v_2\) is spatial rotation in the \((x,y)\) plane, and \(\mathbf v_3,\mathbf v_4,\mathbf v_5\) are translations. The recovered numerical subspaces are consistent with these five reference directions.

Fig.~\ref{fig:RD} visualizes the singular values and the sparsified infinitesimal generators discovered by LieDiscover.

Unlike LieNLSD's fixed library, LieDiscover introduces exponential terms for damped Burgers', trigonometric terms for forced transport, and higher-order polynomial terms for heat and wave. These cases demonstrate explicit open-form symmetry discovery across different equation families.

\begin{table*}[pos=tb]
\begingroup
\centering
\small
\setlength{\tabcolsep}{2pt}
\caption{Long-rollout test NMSE (mean $\pm$ population standard deviation over three runs; lower is better). Bold indicates the lowest mean in each row. All methods are compared at a shared augmentation strength within each equation. GT uses finite sets of analytic generators.}
\label{tab:fno_highlighted}
\begin{tabular}{lccccc}
\toprule
Dataset & FNO + $\emptyset$ & FNO + Ko et al. & FNO + LieNLSD & FNO + LieDiscover & FNO + GT \\
\midrule
Burg.
& $(3.054 \pm 0.4485)\times 10^{-3}$
& $(3.060 \pm 1.018)\times 10^{-3}$
& $(2.897 \pm 0.6780)\times 10^{-3}$
& $\mathbf{(2.028 \pm 0.7338)\times 10^{-3}}$
& $(2.375 \pm 0.4664)\times 10^{-3}$ \\
Heat
& $(6.825 \pm 1.517)\times 10^{-3}$
& $(3.435 \pm 0.5814)\times 10^{-3}$
& $(1.970 \pm 0.2734)\times 10^{-3}$
& $(1.938 \pm 0.4668)\times 10^{-3}$
& $\mathbf{(1.710 \pm 0.3591)\times 10^{-3}}$ \\
KdV
& $(3.236 \pm 0.9397)\times 10^{-1}$
& $(1.721 \pm 0.5513)\times 10^{-1}$
& $(2.024 \pm 0.8190)\times 10^{-1}$
& $\mathbf{(1.521 \pm 0.4959)\times 10^{-1}}$
& $(2.243 \pm 1.228)\times 10^{-1}$ \\
\bottomrule
\end{tabular}
\endgroup

\end{table*}

\subsection{Held-Out Validation}
\label{sec:held_out_validation}
We test whether the recovered generators satisfy the infinitesimal symmetry condition on independent solution trajectories. For each run, LieNLSD and LieDiscover are evaluated on exactly the same 100 primary fitting points and 1,000 held-out points. The held-out trajectories are excluded from both fitting and selection validation. We evaluate all saved generators without refitting, normalizing each coefficient vector $q_a$ to unit Euclidean norm in its original library and keeping it fixed across the two sets.

Using the true benchmark equation $F$, we compute the residual $r_{ia}=J_F(z_i)\Theta_n(z_i;\Theta_f)q_a$ on the saved numerical jets $z_i$. For a set $\mathcal{D}$ and $m$ recovered generators, the reported score is
\begin{equation}
E(\mathcal{D})=\left(\frac{1}{|\mathcal{D}|m}
\sum_{i\in\mathcal{D}}\sum_{a=1}^{m}\|r_{ia}\|_2^2\right)^{1/2}.
\label{eq:heldout_rms}
\end{equation}
The squared norm sums over PDE components. This score measures the residual of each method's recovered set. Fig.~\ref{fig:heldout_comparison} compares training and test residuals, with numerical results in Appendix~\ref{app:held_out_results}.

LieDiscover achieves lower mean held-out residuals on seven of the eight equations. The largest differences occur for forced transport and wave: its mean test residuals are $9.91\times10^{-8}$ and $8.33\times10^{-7}$, respectively, approximately $1.92\times10^4$ and $1.30\times10^4$ times smaller than LieNLSD's. These residuals remain close to their training values even though LieDiscover recovers more generators. On Burgers', heat, and KdV, the mean test residuals are approximately $77$, $191$, and $25$ times smaller, respectively.

Damped Burgers' illustrates the distinction between residual accuracy and recovery coverage. LieNLSD returns two generators with zero residual, whereas LieDiscover recovers three, including the exponential symmetry direction, with a mean test residual of $1.94\times10^{-6}$. The additional direction therefore retains a small residual on independent trajectories.

Reaction--diffusion and Schr\"odinger show larger training-to-test increases for LieDiscover. On reaction--diffusion, its mean residual rises from $1.87\times10^{-4}$ to $5.45\times10^{-3}$, while LieNLSD's rises from $5.73\times10^{-2}$ to $4.18\times10^{0}$. On Schr\"odinger, LieDiscover's mean residual increases from $1.13\times10^{-1}$ to $3.35\times10^{-1}$, compared with LieNLSD's test value of $5.28\times10^{-1}$. Thus, the advantage persists on these equations, but the Schr\"odinger results show a smaller separation between methods and a higher residual than the other LieDiscover results.

\subsection{PDE Solving}
\label{sec:Application}
For downstream PDE solving, we evaluate Burgers', heat, and KdV, matching the three-equation scope of LieNLSD~\citep{hu2025explicit}. We compare FNO without augmentation with FNO trained using LPSDA~\citep{brandstetter2022lie} and generators from \citet{ko2024learning}, LieNLSD, LieDiscover, or the analytic ground-truth (GT) reference. Each method uses its complete generator set. Table~\ref{tab:fno_highlighted} reports the final long-rollout test NMSE as the mean and population standard deviation over three FNO training seeds.

For each training seed, the reported NMSE follows the implementation's time-summed normalization:
\[
\mathrm{NMSE}
=\frac{1}{B_{\mathrm{eval}}}
\sum_{i\in\mathcal D_{\mathrm{test}}}
\sum_{t\in\mathcal T_{\mathrm{pred}}}
\frac{\|\widehat{\boldsymbol u}_i(t)-\boldsymbol u_i(t)\|_2^2}
{\|\boldsymbol u_i(t)\|_2^2},
\]
where $\boldsymbol u_i(t)$ contains the ten retained spatial values, and $\mathcal D_{\mathrm{test}}$ consists of two test trajectories evaluated in a single batch. The predicted times $\mathcal T_{\mathrm{pred}}$ are covered by an autoregressive rollout using 20 history steps and 20 prediction steps per block, after skipping the initial 20 steps. The same convention is used for all methods.

LieDiscover reduces the mean NMSE relative to unaugmented FNO on all three equations, by $33.6\%$ on Burgers', $71.6\%$ on heat, and $53.0\%$ on KdV. It achieves the lowest mean among the five methods on Burgers' and KdV. On heat, the GT reference has the lowest mean, while LieDiscover has the lowest mean among the methods using learned or discovered generators.

Relative to LieNLSD, LieDiscover reduces the mean NMSE by $30.0\%$, $1.6\%$, and $24.9\%$ on Burgers', heat, and KdV, respectively. These results support the usefulness of the discovered generators for downstream prediction at the reported settings.

\section{Further Analysis}
We study the contributions of the main components, robustness to observation noise, and sources of residual error. Ablation experiments compare Jacobian computation, dimension selection, and library search strategies to assess their roles in symmetry recovery. Noise experiments test how well the recovered symmetries persist as observation noise increases. Error analysis compares fitted and true equations with discovered and analytic generators to examine how these choices affect the measured residuals.

\subsection{Ablation Studies}

We conduct three ablation experiments to isolate the contributions of the main components in LieDiscover. First, we compare the WSINDy-based Jacobian computation with the neural surrogate used by LieNLSD in terms of equation residuals and Jacobian errors on all eight PDEs. Second, we evaluate the proposed dimension-selection rule against a fixed singular-value cutoff and the largest-gap heuristic over certified candidate libraries. Third, we compare the learned search policy with fixed-library enumeration and Monte Carlo sampling under matched data, budgets, and post-processing. Together, these experiments assess the effects of explicit equation identification, validated rank selection, and policy-guided library construction.

\subsubsection{Jacobian Computation}
\label{sec:jacobian_ablation}

We compare our WSINDy-based Jacobian with the neural surrogate used by LieNLSD \citep{hu2025explicit}. LieNLSD fits an MLP $f_\phi(z)$ to the temporal derivative $y$ using mean-squared error. The network has three hidden layers of width 200 and Sigmoid activations. Our method fits the explicit equation $\widehat f(z)$ by weak sparse regression, as described in Section~\ref{sec:jacobian}. Both methods use automatic differentiation to compute
\begin{equation}
    J_{\widehat F}(z)=[D_z\widehat f(z),-I_q],
    \qquad \widehat F(z,y)=\widehat f(z)-y,
\end{equation}
where $\widehat f=f_\phi$ for LieNLSD. We use the official LieNLSD training settings for its six PDE benchmarks. The same MLP architecture and Burgers' training schedule are used for damped Burgers' and forced transport.

Each frozen model is evaluated on the same 1,000 test points per equation. We sample points evenly across test trajectories with three seeds.

For output component $a$, let $\widehat{\boldsymbol f}_a$ and $\boldsymbol y_a$ collect the predicted and observed temporal derivatives over the test points. Let $\widehat G_a$ and $G_a^\ast$ stack the rows $D_z\widehat f_a(z_i)$ and $D_z f_a^\ast(z_i)$, respectively. The reference $f^\ast$ is the analytic benchmark equation. We report
\begin{equation}
\begin{aligned}
    E_F&=\left(\frac{1}{q}\sum_{a=1}^{q}
       \frac{\|\widehat{\boldsymbol f}_a-\boldsymbol y_a\|_2^2}
            {\|\boldsymbol y_a\|_2^2}\right)^{1/2},\\
    E_J&=\left(\frac{1}{q}\sum_{a=1}^{q}
       \frac{\|\widehat G_a-G_a^\ast\|_F^2}
            {\|G_a^\ast\|_F^2}\right)^{1/2}.
\end{aligned}
\label{eq:jacobian_ablation_metrics}
\end{equation}
$E_F$ measures the equation residual against the saved numerical temporal derivatives. The Jacobian error $E_J$ measures the $D_zf$ block, the remaining $-I_q$ block is exact in both methods.

\begin{table}[pos=t]
\centering
\caption{Equation residuals $E_F$ and Jacobian errors $E_J$ on 1,000 test points per equation. Lower is better. Bold entries mark the smaller error.}
\label{tab:jacobian_ablation}
\begingroup
\setlength{\tabcolsep}{3pt}
\renewcommand{\arraystretch}{1.12}
\resizebox{\columnwidth}{!}{%
\begin{tabular}{lcccc}
\toprule
& \multicolumn{2}{c}{$E_F$} & \multicolumn{2}{c}{$E_J$} \\
\cmidrule(lr){2-3}\cmidrule(lr){4-5}
Equation & MLP& WSINDy & MLP& WSINDy \\
\midrule
Burgers' & $8.30\!\times\!10^{-3}$ & $\mathbf{2.85\!\times\!10^{-3}}$ & $4.59\!\times\!10^{-2}$ & $\mathbf{5.72\!\times\!10^{-5}}$ \\
Heat & $2.14\!\times\!10^{-3}$ & $\mathbf{3.03\!\times\!10^{-5}}$ & $3.26\!\times\!10^{-3}$ & $\mathbf{1.97\!\times\!10^{-6}}$ \\
KdV & $8.32\!\times\!10^{-3}$ & $\mathbf{4.45\!\times\!10^{-4}}$ & $2.05\!\times\!10^{-2}$ & $\mathbf{1.01\!\times\!10^{-5}}$ \\
Wave & $1.13\!\times\!10^{-3}$ & $\mathbf{5.68\!\times\!10^{-8}}$ & $4.81\!\times\!10^{-3}$ & $\mathbf{2.01\!\times\!10^{-7}}$ \\
Reaction--diffusion & $2.24\!\times\!10^{-2}$ & $\mathbf{2.02\!\times\!10^{-3}}$ & $5.23\!\times\!10^{-2}$ & $\mathbf{6.90\!\times\!10^{-5}}$ \\
Schr\"odinger & $9.38\!\times\!10^{-3}$ & $\mathbf{1.86\!\times\!10^{-3}}$ & $1.03\!\times\!10^{-2}$ & $\mathbf{3.76\!\times\!10^{-4}}$ \\
Damped Burgers' & $4.98\!\times\!10^{-3}$ & $\mathbf{5.36\!\times\!10^{-5}}$ & $4.77\!\times\!10^{-3}$ & $\mathbf{7.32\!\times\!10^{-6}}$ \\
Forced transport & $2.82\!\times\!10^{-3}$ & $\mathbf{2.41\!\times\!10^{-8}}$ & $6.33\!\times\!10^{-3}$ & $\mathbf{3.53\!\times\!10^{-8}}$ \\
\bottomrule
\end{tabular}%
}
\endgroup
\end{table}

Table~\ref{tab:jacobian_ablation} shows lower errors for WSINDy on all eight equations. Its Jacobian error ranges from $3.53\times10^{-8}$ to $3.76\times10^{-4}$. LieNLSD's error ranges from $3.26\times10^{-3}$ to $5.23\times10^{-2}$. These results support using the identified explicit equation to compute the Jacobians required by symmetry discovery.

\begin{figure}[pos=t]
    \centering
    \includegraphics[width=\columnwidth]{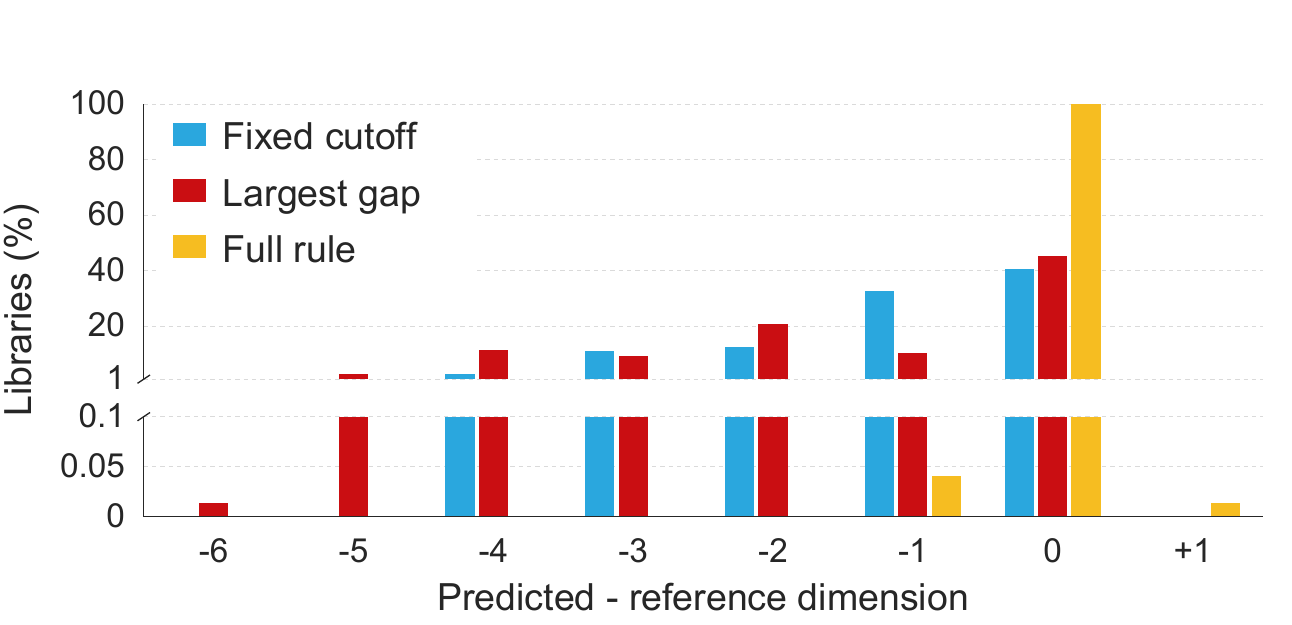}
    \caption{Distribution of the signed dimension error $\Delta d=\widehat d-d_{\mathrm{ref}}$ over the $7371$ certified retained libraries. The vertical axis is broken between $0.1\%$ and $1\%$. Missing bars indicate zero observations. }
    \label{fig:svd_ablation}
\end{figure}

\begin{figure}[pos=tb]
    \centering
    \includegraphics[width=\columnwidth]{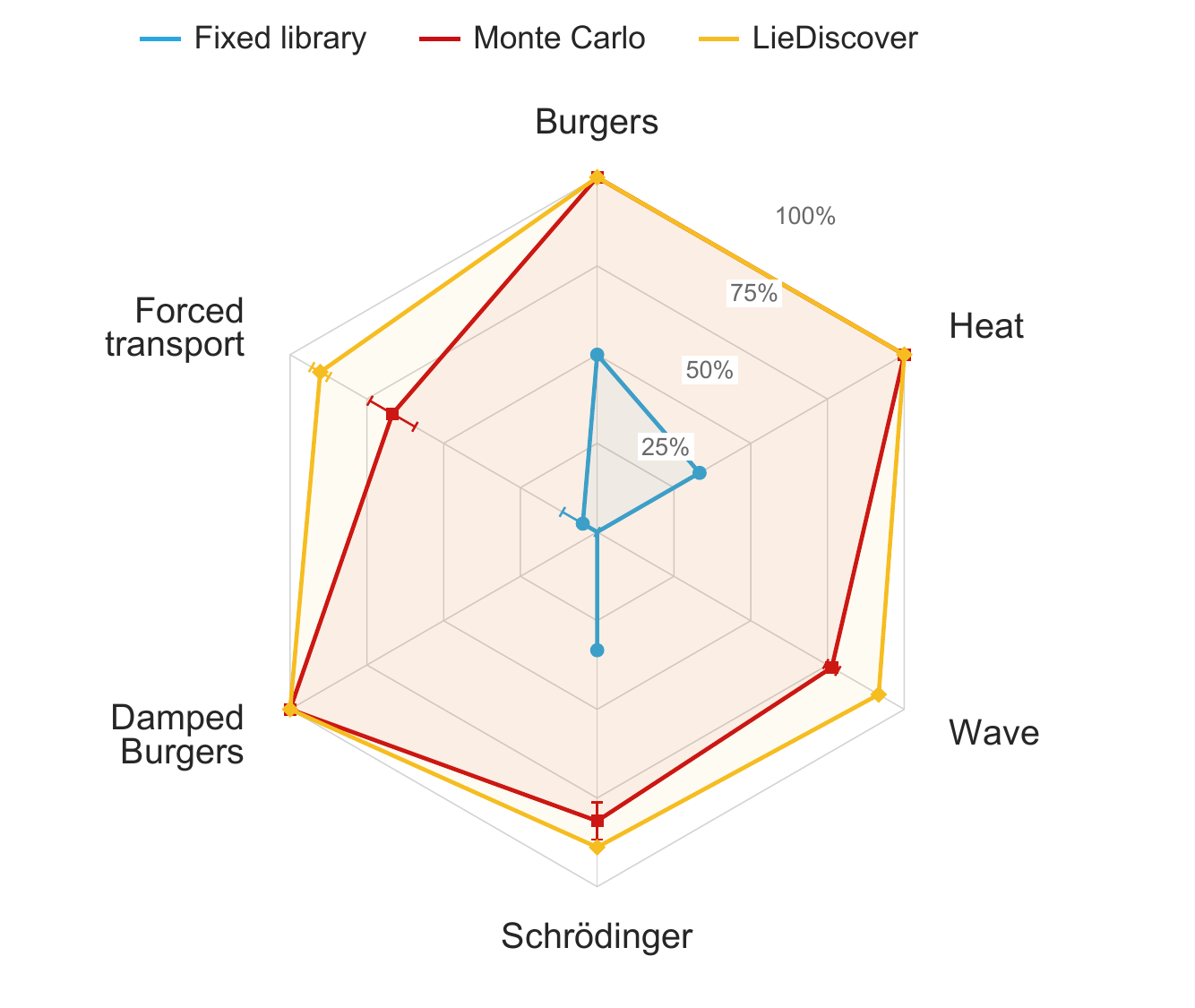}
    \caption{Search ablation on six equations where library expansion can add symmetry directions. Values show the recovery ratio $R$ as a percentage. Curves and error bars show the mean and one population standard deviation over three seeds. Lower error bars stop at zero. Fixed library uses the full-pool capacity. Monte Carlo and LieDiscover use the $K$-step capacity. Runs with no accepted output contribute zero recovery.}
    \label{fig:search_ablation}
\end{figure}

\begin{figure*}[pos=tb]
    \centering
    \includegraphics[width=\linewidth]{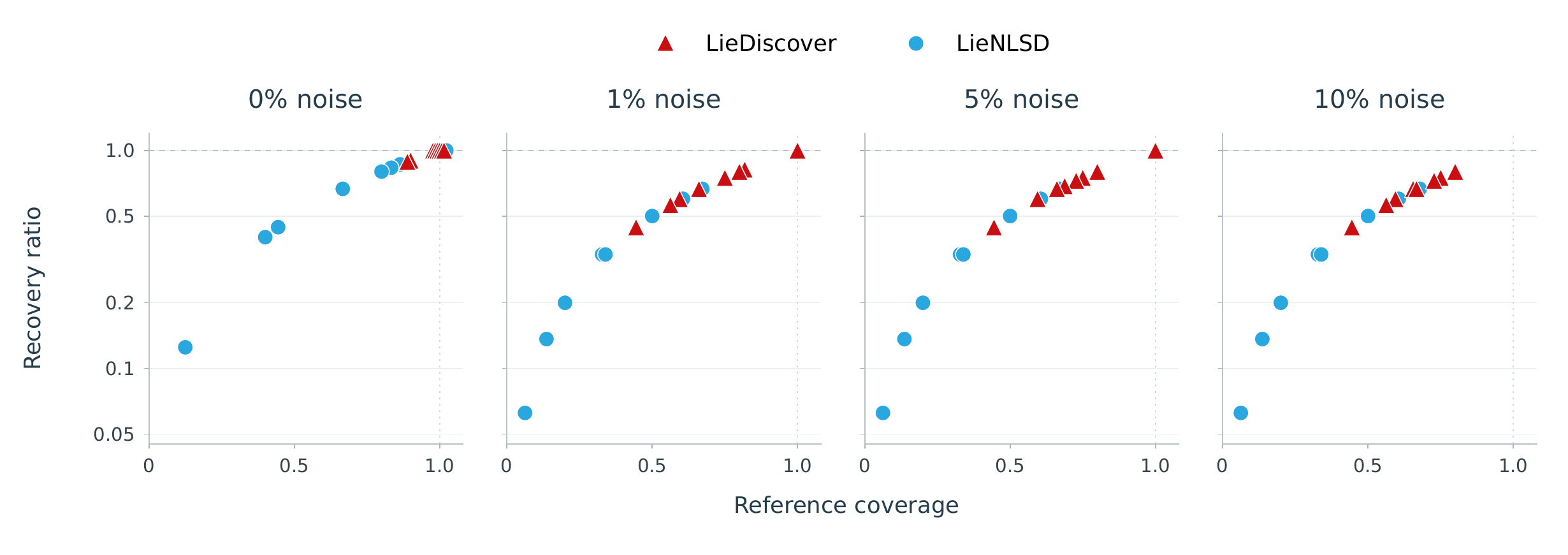}
    \caption{Reference recovery under observation noise. Each panel contains eight PDE results per method. Red triangles denote LieDiscover; blue circles denote LieNLSD. Both axes show $R_{\mathrm{noise}}$, with a logarithmic vertical scale. Small horizontal offsets distinguish overlapping markers.}
    \label{fig:noise_recovery}
\end{figure*}

\subsubsection{Dimension Selection}
\label{sec:svd_ablation}

We randomly construct 7,371 unique candidate function libraries, whose symmetry dimensions are certified using exact determining equations. We consider three variants: \textbf{Fixed cutoff}, which selects the dimension using $\sigma_i/\sigma_1\leq 10^{-6}$; \textbf{Largest gap}, which selects the dimension corresponding to the largest positive log gap; and \textbf{Full rule}, which generates candidates from the three largest positive log gaps across three fitting subsets and selects the final dimension according to Eq.~\eqref{eq:validated_dimension}.

The fixed-cutoff and largest-gap rules correctly classify $2987$ ($40.52\%$) and $3331$ ($45.19\%$) libraries. All their errors are underestimates. The full rule correctly classifies $7367$ ($99.95\%$) libraries, with three one-dimensional underestimates and one one-dimensional overestimate. These scores are conditioned on the retained library. The retained Schr\"odinger libraries support only six to eight directions, so its score does not imply recovery of the full nine-dimensional algebra.

Fig.~\ref{fig:svd_ablation} shows that the fixed-cutoff and largest-gap rules often miss several symmetry directions. The full rule concentrates almost all predictions at zero error, and its four errors differ from the reference by only one dimension. This comparison shows that residual and stability checks help distinguish valid null directions from misleading spectral gaps.

\subsubsection{Search Strategy}
\label{sec:search_ablation}

We compare three ways to construct the function library. \textbf{Fixed library} enumerates all functions allowed by the same expression grammar and length constraints as LieDiscover. This includes the permitted non-polynomial terms, such as exponential functions for damped Burgers\textquoteright{} and trigonometric functions for forced transport. It computes a basis from the complete candidate library in a single solve. \textbf{Monte Carlo} samples an available expression length uniformly, then samples a function uniformly within that length. It uses the same initial library, trajectory count, $K$-step budget, acceptance rules, and final selection procedure as LieDiscover. \textbf{LieDiscover} proposes expressions with the learned policy. 

All three methods share prolonged-column normalization, dimension selection, sparsification, and pruning. For each seed, they use the same WSINDy equation, three fitting groups of 100 points, a 1,000-point scale pool, and a 1,000-point validation pool. Each method is evaluated over three seeds. We measure the recovery ratio $R=d_{\mathrm{val}}/D$. The denominator $D$ is the largest true dimension representable under the method's library budget.

Let $\mathcal L$ be the complete candidate pool, including $\Theta_0$. Let $\nu(\Theta)$ denote the true symmetry dimension expressible by $\Theta$. Fixed library uses $D=D_{\mathrm{full}}=\nu(\mathcal L)$. Monte Carlo and LieDiscover use $D=D_K$, where
\begin{equation}
    D_K=\max_{\substack{\Theta_0\subseteq\Theta\subseteq\mathcal L\\
                       |\Theta\setminus\Theta_0|\le K}}\nu(\Theta).
    \label{eq:search_capacity}
\end{equation}
We compute these capacities from exact determining equations and the expression budget. Table~\ref{tab:search_ablation_all} gives the denominators and results for all eight equations. Fig.~\ref{fig:search_ablation} shows the six cases with $D_K>\nu(\Theta_0)$. The initial libraries already attain $D_K$ for KdV and reaction--diffusion.

Although the fixed library includes the permitted non-polynomial terms, it returns no accepted output on wave, reaction--diffusion, or damped Burgers' under the shared numerical selection procedure (Table~\ref{tab:search_ablation_all}). These results highlight two benefits of dynamic library construction: introducing the required functions and avoiding the difficulties of solving for generators from a large candidate library all at once.

Monte Carlo and LieDiscover both reach $100\%$ on Burgers', heat, and damped Burgers'. LieDiscover raises recovery from $76.39\%$ to $91.67\%$ on wave, from $81.48\%$ to $88.89\%$ on Schr\"odinger, and from $66.67\%$ to $90.20\%$ on forced transport. These gains show the benefit of learned proposals on tasks that require more extensive library construction.

\subsection{Robustness Analysis}
\label{sec:noise_robustness}

We evaluate both methods on all eight PDEs with Gaussian noise added to the observed grid values. The noise standard deviation is $1\%$, $5\%$, or $10\%$ of each field's standard deviation within a trajectory. Both methods use the same noisy observations.

LieDiscover fits a coordinate MLP to the noisy fields and computes derivatives by automatic differentiation. It then identifies the governing equation with WSINDy and searches for its symmetry generators. The WSINDy settings are unchanged from the clean experiments. LieNLSD uses finite-difference derivatives to train its governing-relation surrogate and discovers generators with its fixed candidate library.

For each PDE, we use the fixed analytic reference set described in Appendix~\ref{app:four_group_residuals}. A reference direction is counted as recovered when its angle to the discovered generator space is at most $0.1$ radians in the common coefficient representation. The recovery ratio is $R_{\mathrm{noise}}=k/K$, where $k$ is the number of recovered reference directions and $K$ is the reference-set size.

Fig.~\ref{fig:noise_recovery} shows that LieDiscover recovers more reference directions in $23$ of the $24$ noisy comparisons and ties LieNLSD in the remaining case. The largest gains occur on wave and forced transport. At $10\%$ noise, LieDiscover recovers $16$ of the $22$ wave references and $9$ of the $16$ forced-transport references. LieNLSD recovers $3$ and $1$, respectively. On heat, LieDiscover retains $6$ of the $10$ references at all three noise levels, compared with $2$ for LieNLSD.

LieNLSD recovers only translation directions in the noisy experiments. LieDiscover also recovers non-translation directions on all eight PDEs at $1\%$ and $5\%$ noise, and on seven PDEs at $10\%$ noise. Damped Burgers' is the remaining case: LieDiscover recovers all three references at the two lower noise levels and retains the two translations at $10\%$. Appendix~\ref{app:noise_recovery} gives the per-equation counts and separates translation from non-translation directions.

\subsection{Error Analysis}
\label{sec:four_group_residuals}
\begin{figure*}[pos=tb]
\centering
\includegraphics[width=\linewidth]{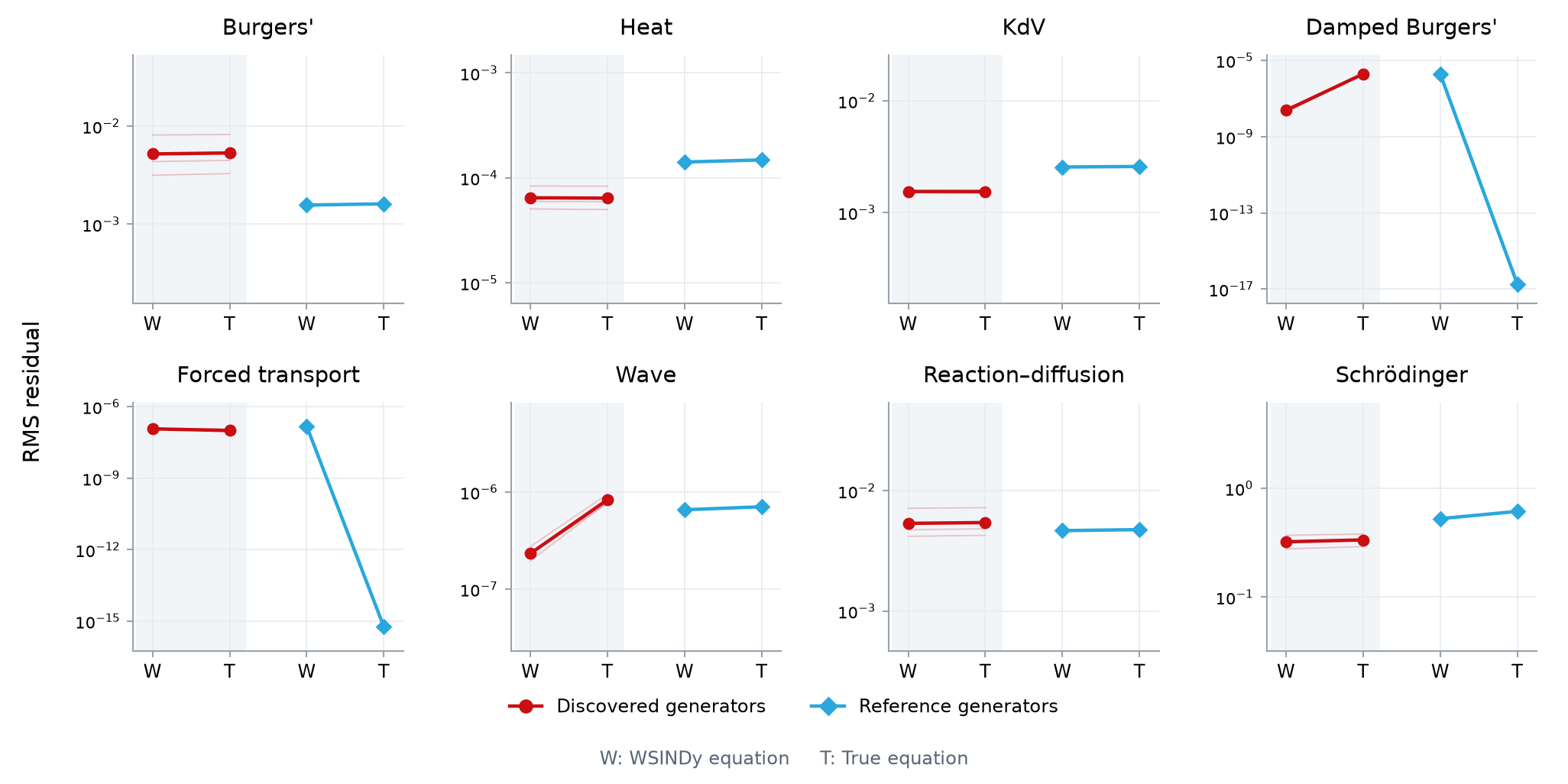}
\caption{Symmetry residuals of discovered and reference generators evaluated using the WSINDy and true equations on common held-out numerical jets. Red circles denote discovered generators, and blue diamonds denote analytic reference generators. Solid markers show means over three search seeds.}
\label{fig:four_group_residuals}
\end{figure*}

We examine how the fitted equation affects the residual assessment in the experiments without added observation noise. Let $\widehat F$ be the frozen WSINDy equation, $F$ the true equation, $\widehat{\mathcal V}$ the discovered generators, and $\mathcal V^\star$ a fixed analytic reference set. Writing $E(G,\mathcal V)$ for the held-out score in Eq.~\eqref{eq:heldout_rms} evaluated with equation $G$ and generator set $\mathcal V$, we compare
\[
E(\widehat F,\widehat{\mathcal V}),\quad E(F,\widehat{\mathcal V}),\quad
E(\widehat F,\mathcal V^\star),\quad E(F,\mathcal V^\star).
\]
All four evaluations use the exact same 1,000 held-out points and saved numerical jets as Section~\ref{sec:held_out_validation}. The complete numerical results are given in Appendix~\ref{app:four_group_residuals}.

As shown in Fig.~\ref{fig:four_group_residuals}, the discovered generators generally exhibit consistent residuals when evaluated using the WSINDy-identified and true equations. The differences are small for most benchmarks, while more noticeable changes occur for damped Burgers' and wave equations. For forced transport, the residuals even decrease when switching to the true equation. The analytic reference generators (blue diamonds) also show that the two equations can have substantially different residuals on the same numerical data, particularly for damped Burgers' and forced transport. 
These results suggest that the discovered generators are generally stable to changes in the evaluation equation.

\section{Discussion}
\label{sec:discussion}

LieDiscover demonstrates the ability to recover explicit symmetry generators across eight PDE benchmarks, including higher-order polynomial and transcendental forms. Nevertheless, complete recovery of the reference symmetry subspace remains challenging. The nonlinear Schr\"odinger diagnostic in Section~\ref{sec:schrodinger_results} shows that the current stability checks can reject a candidate subspace containing all reference directions, even when its residual is below the acceptance tolerance. This illustrates the sensitivity of numerical dimension estimation to data sampling. Improving the reliability of dimension estimation and developing more effective criteria to distinguish genuine symmetries from numerical near-null directions remain important directions for future work.

Another limitation of the current framework is its reliance on identifying an explicit governing equation from observational data before computing its Jacobian. Although this strategy allows us to exploit the linear structure of the infinitesimal symmetry criterion, the accuracy of symmetry discovery depends on the quality of the identified equation. Errors in equation identification may propagate to the computed Jacobians and subsequently affect symmetry dimension estimation and generator recovery. Future work will investigate joint equation and symmetry discovery, where the governing equation and its associated symmetry structure are identified collaboratively, allowing symmetry information to provide additional constraints on equation identification rather than treating the two tasks as strictly separate stages.

Finally, although LieDiscover dynamically constructs symbolic expressions rather than selecting terms from a predefined candidate library, its search remains constrained by expression complexity and a finite search budget. Future research will explore adaptive grammar expansion and more flexible search strategies to improve the expressiveness of symbolic discovery. In addition, incorporating physical priors, such as domain-specific constraints, could guide expression construction, reduce the search space, and improve the reliability of the discovered generators.

\section{Conclusion}
In this work, we propose LieDiscover, a reinforcement learning framework for discovering symmetries from data. Unlike existing methods that rely on a predefined library, LieDiscover can explicitly recover open-form generators, including high-order polynomials and transcendental functions. It uses an encoder-decoder network as an actor to build symbolic expressions, while the critic estimates state values to guide the process. At each search step, adaptive SVD proposes possible symmetry subspace dimensions from spectral gaps. Residual and stability checks screen the corresponding subspaces. The resulting dimension estimates guide library expansion through step-wise rewards. Across eight PDE benchmarks, LieDiscover recovers explicit generators with polynomial, exponential, and trigonometric terms. Under observation noise, it recovers more reference symmetry directions in 23 of 24 comparisons. Data augmentation with the recovered generators also lowers mean FNO prediction error on all three downstream tasks. Future work will extend LieDiscover to discrete symmetries and more complex dynamical systems.

\appendix
\section*{Appendix}

\section{Reward Function Analysis}
\label{app:proof}
The reward function is designed to prioritize the discovery of additional independent symmetry generators while accounting for search delay and expression complexity. Since these factors may compete with the dimensional improvement, we analyze whether a larger dimension increment can consistently yield a higher reward despite variations in delay and complexity. The following theorem establishes a sufficient condition under which the dimensional improvement dominates these two penalty factors, ensuring the desired reward ordering for all admissible pairs of accepted transitions.
\begin{theorem}
\label{thm:reward_condition_appendix}
Let $A$ and $B$ be two accepted transitions with dimensional increments such that $\bar{d} \ge \Delta d_A > \Delta d_B \ge 1$. Assume $0<\eta\le1$, $\alpha>0$, and $\lambda>0$. Their delays satisfy $0\le h_A,\allowbreak h_B\le K-1$, and their expression complexities satisfy $c_A,\allowbreak c_B\in[c_{\min},\allowbreak c_{\max}]$. A sufficient condition for $r_A>r_B$ to hold for every admissible pair is
\begin{equation}
\label{eq:theorem_condition_appendix}
\bar{d}\eta^{K-1} - (\bar{d}-1) > \alpha \left( e^{-\lambda c_{\min}} - e^{-\lambda c_{\max}} \right),
\end{equation}
where $\bar{d}$ is the finite upper bound of one-step dimensional increments.
\end{theorem}

\begin{proof}
Define
\begin{equation*}
    \delta_c = e^{-\lambda c_{\min}} - e^{-\lambda c_{\max}}.
\end{equation*}
Since $c_{\min} \leq c_{\max}$ and $\lambda > 0$, we have $\delta_c \geq 0$.

For brevity, write
\begin{equation*}
    r_A = r(\Delta d_A, h_A, c_A), \quad r_B = r(\Delta d_B, h_B, c_B).
\end{equation*}
Their difference is
\begin{equation*}
    r_A - r_B = \Delta d_A\eta^{h_A} - \Delta d_B\eta^{h_B} + \alpha \left( e^{-\lambda c_A} - e^{-\lambda c_B} \right).
\end{equation*}

Let $H=K-1$. Because $0 < \eta \leq 1$ and $0 \leq h_A,\allowbreak  h_B \leq H$,
\begin{equation*}
    \eta^{h_A} \geq \eta^H, \quad \eta^{h_B} \leq 1.
\end{equation*}
Moreover, since $c \mapsto e^{-\lambda c}$ is decreasing,
\begin{equation*}
    e^{-\lambda c_A} \geq e^{-\lambda c_{\max}}, \quad e^{-\lambda c_B} \leq e^{-\lambda c_{\min}}.
\end{equation*}
Consequently,
\begin{equation}
\label{eq:reward_difference_lower_bound}
    r_A - r_B \geq \Delta d_A\eta^H - \Delta d_B - \alpha\delta_c.
\end{equation}

The dimension increments are integers and satisfy $\Delta d_A > \Delta d_B$. Hence,
\begin{equation*}
    \Delta d_A \geq \Delta d_B + 1.
\end{equation*}
Applying this inequality to Eq.~\eqref{eq:reward_difference_lower_bound} gives
\begin{align*}
    r_A - r_B &\geq (\Delta d_B + 1)\eta^H - \Delta d_B - \alpha\delta_c \\
    &= \eta^H + \Delta d_B(\eta^H - 1) - \alpha\delta_c.
\end{align*}

Since $\eta^H - 1 \leq 0$ and $\Delta d_B \leq \bar d - 1$, it follows that
\begin{equation*}
    \Delta d_B(\eta^H - 1) \geq (\bar d - 1)(\eta^H - 1).
\end{equation*}
Therefore,
\begin{align*}
    r_A - r_B &\geq \eta^H + (\bar d - 1)(\eta^H - 1) - \alpha\delta_c \\
    &= \bar d\eta^H - (\bar d - 1) - \alpha\delta_c.
\end{align*}

By Eq.~\eqref{eq:theorem_condition_appendix}, the final expression is strictly positive. Thus,
\begin{equation*}
    r_A - r_B > 0,
\end{equation*}
and hence $r_A > r_B$ for every admissible pair of transitions.
\end{proof}

In our experiments, we use the practical approximation $\eta > \sqrt[K-1]{(\bar{d}-1)/\bar{d}}$ with $\bar d=p+q$. Adding one scalar expression introduces $p+q$ coefficient columns, which bounds the increase in exact nullity. With the dimension selected numerically for each library, however, this bound need not hold for the estimated nullity. This parameter choice also omits the correction induced by the simplicity bonus because its weight $\alpha$ is small. It is therefore a heuristic for emphasizing dimension improvement, rather than a universal guarantee of reward ordering. The theorem applies when its increment bound and full sufficient condition are satisfied.

\section{Detailed Algorithm}
\label{app:algorithms}
\subsection{The Construction of \texorpdfstring{$\Theta_n$}{Theta n}}
In LieNLSD~\citep{hu2025explicit}, the prolonged library \(\Theta_n\) is manually derived and fixed for a pre-defined candidate library \(\Theta\). In contrast, LieDiscover adaptively expands \(\Theta\) during the search process, so the environment must automatically derive the corresponding \(\Theta_n\) for each updated library.
\begin{theorem}[Prolongation in library coordinates]
\label{thm:prolongation}
Let $\Theta(x,u)\in\mathbb R^r$ be a smooth function library and let $W\in\mathbb R^{(p+q)\times r}$ be constant. Write the point vector field as
\[
\begin{aligned}
\mathbf v=\sum_{i=1}^p\xi^i\partial_{x^i}
+\sum_{\alpha=1}^q\eta^\alpha\partial_{u^\alpha},&\\
\xi^i=\Theta^\top W_i,\qquad
\eta^\alpha=\Theta^\top W_{p+\alpha},
&
\end{aligned}
\]
where $W_i\in\mathbb R^r$ is the transpose of row $i$ of $W$. Its $n$-th prolongation is
\begin{equation}
\label{eq:eq_1}
\operatorname{pr}^{(n)}\mathbf v
=\mathbf v+\sum_{\alpha=1}^q\sum_{1\le |J|\le n}
\phi_\alpha^J\partial_{u_J^\alpha},
\end{equation}
where the sum includes each distinct symmetric derivative multi-index once. By the general prolongation formula \citep[Section~2.3]{olver1993applications},
\begin{equation}
\label{eq:eq_2}
\phi_\alpha^J
=D_J\!\left(\eta^\alpha-\sum_{i=1}^p\xi^i u_i^\alpha\right)
+\sum_{i=1}^p\xi^i u_{J,i}^\alpha.
\end{equation}
Here $D_J=D_{j_1}\cdots D_{j_k}$ is a composition of total derivatives and $u_{J,i}^\alpha=D_i u_J^\alpha$. Equivalently, the coefficients satisfy
\begin{equation}
\label{eq:prolongation_recursion}
\phi_\alpha^{Jk}=D_k\phi_\alpha^J
-\sum_{i=1}^p u_{J,i}^\alpha D_k\xi^i,
\qquad \phi_\alpha^\emptyset=\eta^\alpha.
\end{equation}
Thus every prolonged coefficient remains linear in $W$, yielding the matrix representation in Eq.~\eqref{eq:infinitesimal generator}.
\end{theorem}

Algorithms~\ref{alg:theta_n_construction} and~\ref{alg:recursive_coeff} implement Eq.~\eqref{eq:prolongation_recursion} by constructing coefficient blocks and then substituting the total derivatives of the current library. For a coordinate target $x^i$, the corresponding component is $\xi^i$; no derivative recursion is needed. The target order matches the Jacobian columns. This includes explicit coordinate dependence, such as the forcing in the transport equation.

\begin{algorithm}[!htbp]
\caption{Recursive Construction of $\Theta_n$ from $\Theta$}
\label{alg:theta_n_construction}
\compactalgorithm
\begin{algorithmic}[1]
\REQUIRE Library $\Theta$, independent variables $\mathbf{x}$, dependent variables $\mathbf{u}$, ordered coordinate and jet targets $\mathcal{Q}$
\ENSURE Prolonged matrix $\Theta_n$
\STATE Initialize an empty matrix $\Theta_n$.
\FOR{each target $z\in\mathcal{Q}$, in order}
    \IF{$z=x^i$ is an independent coordinate}
        \STATE Append a row with $\Theta^\top$ in the $x^i$ block and zeros in all other blocks.
    \ELSE
        \STATE Write $z=u^\alpha_J$ and compute $(\{W_{x_i}\},W_{u^\alpha})\leftarrow\textsc{RecursiveCoeff}(u^\alpha,J)$.
        \STATE Instantiate all abstract derivatives $D_I\Theta^\top$ using the current library $\Theta$.
        \STATE Append a row in $(x^1,\ldots,x^p,u^1,\ldots,u^q)$ block order, using zero blocks for $u^\beta\ne u^\alpha$.
    \ENDIF
\ENDFOR
\RETURN $\Theta_n$
\end{algorithmic}
\end{algorithm}

\begin{algorithm}[!htbp]
\caption{RecursiveCoeff$(u^\alpha,\allowbreak J)$}
\label{alg:recursive_coeff}
\compactalgorithm
\begin{algorithmic}[1]
\REQUIRE Dependent variable $u^\alpha$, derivative multi-index $J$
\ENSURE Abstract coefficients $\{W_{x_i}\}$ and $W_{u^\alpha}$

\IF{$J=\emptyset$}
    \RETURN $(\{0\}_{i=1}^p,\Theta^\top)$.
\ENDIF
\STATE Split $J=J'k$; set $(\{W'_{x_i}\},W'_{u^\alpha})\leftarrow\textsc{RecursiveCoeff}(u^\alpha,J')$.
\STATE $W_{u^\alpha}\leftarrow D_kW'_{u^\alpha}$.
\STATE $W_{x_i}\leftarrow D_kW'_{x_i}-u^\alpha_{J'i}D_k\Theta^\top$ for all $i=1,\ldots,p$.
\RETURN $(\{W_{x_i}\},W_{u^\alpha})$.
\end{algorithmic}
\end{algorithm}

\subsection{Adaptive Singular Value Decomposition}
\label{app:adaptive_svd}
SVD estimates the numerical null space of the linear constraint matrix. A fixed absolute cutoff, such as that used by LieNLSD~\citep{hu2025explicit}, is sensitive to matrix scale. The largest spectral gap alone can also miss approximate symmetry directions. We therefore use multiple spectral gaps to construct a set of candidate dimensions, which are then screened using residual and stability criteria.

For a library with $N_c$ coefficient columns, let $\bar C_b$ stack $J_i\bar B_i$ for $i\in\mathcal I_b$, where $b=0,1,2,3$. Thus, $\bar C_0$ contains the reference-pool blocks. Let $\bar C_v$ stack the validation blocks. All matrices use the scales computed on $\mathcal I_0$. For each $b=1,2,3$, we decompose $\bar C_b^\top\bar C_b=V_b\Sigma_b^2V_b^\top$ and sort the recovered singular values in descending order. With numerical floor $\delta=10^{-12}$, define
\begin{equation}
\label{eq:dimension_gap_candidates}
\begin{aligned}
    \bar\sigma_{b,j}&=\max(\sigma_{b,j}/\sigma_{b,1},\delta),\\
    g_{b,j}&=\log_{10}(\bar\sigma_{b,j}/\bar\sigma_{b,j+1}),\\
    \mathcal N&=\bigcup_{b=1}^{3}\{N_c-j:j\in\mathcal J_b\},
\end{aligned}
\end{equation}
where $\mathcal J_b$ contains the indices of up to three largest positive gaps. Ties are resolved by increasing index. There are no prescribed bounds on the spectral cutoff. The floor stabilizes spectral ratios. It does not enter the validation checks or condition number.

For every $d\in\mathcal N$, let $Q_{b,d}$ contain the last $d$ columns of $V_b$. The reference-pool residual, validation residual, and stability error are
\begin{equation}
\label{eq:dimension_validation}
\begin{aligned}
    e_0(d)&=\max_{b=1,2,3}
    \frac{\|\bar C_0Q_{b,d}\|_2}{\|\bar C_0\|_2},\\
    e_R(d)&=\max_{b=1,2,3}
    \frac{\|\bar C_vQ_{b,d}\|_2}{\|\bar C_v\|_2},\\
    e_P(d)&=\max_{1\le b<c\le3}
    \|Q_{b,d}Q_{b,d}^{\top}-Q_{c,d}Q_{c,d}^{\top}\|_2.
\end{aligned}
\end{equation}
The first two checks measure the largest residual over unit coefficient vectors in each candidate space. The reference pool remains on the training side. The validation set is checked separately. The third check compares coefficient spaces without depending on the chosen bases. A small residual alone need not imply an accurate singular subspace \citep{park2023fast}. Subspace perturbation bounds depend on both the perturbation size and spectral separation \citep{wedin1972perturbation}. This motivates checking agreement across sample subsets.

Small coefficient residuals may also arise when feature combinations nearly cancel. We therefore compare the generator values on the shared reference pool. For $M=|\mathcal I_0|$, define
\begin{equation}
\label{eq:generator_value_matrix}
\begin{aligned}
    H_i&=(I_{p+q}\otimes\Theta(x_i,u_i)^\top)S^{-1},\\
    H&=M^{-1/2}\operatorname{stack}_{i\in\mathcal I_0}H_i.
\end{aligned}
\end{equation}
Here, $\operatorname{stack}$ joins the blocks vertically. The matrix $H$ uses the original features and the same column scales as $\bar C_b$. It needs no new samples, derivatives, or Jacobians.

For each fitted basis, compute a thin SVD of $Y_{b,d}=HQ_{b,d}$. Its numerical rank must equal $d$. We count singular values above
\begin{equation}
\label{eq:generator_function_rank}
    \tau_{b,d}=\max\{M(p+q),d\}\,
    \epsilon_{\mathrm{mach}}\,\sigma_1(Y_{b,d}),
\end{equation}
where $\epsilon_{\mathrm{mach}}\approx2.22\times10^{-16}$ is the machine epsilon for double precision. If any rank is below $d$, set $e_X(d)=+\infty$. Otherwise, let $U_{b,d}$ contain the $d$ orthonormal left singular vectors and define
\begin{equation}
\label{eq:generator_function_stability}
    e_X(d)=\max_{1\le b<c\le3}
    \|U_{c,d}-U_{b,d}(U_{b,d}^\top U_{c,d})\|_2.
\end{equation}
This compares the spaces spanned by the generator functions. For equal dimensions, it equals the projector distance in Lemma~\ref{lem:pruning_subspace_distance} \citep[Lemma~1.5]{drineas2019low}. Thus, both $e_P$ and $e_X$ measure the sine of the largest principal angle. The map $H$ can amplify subspace changes when feature combinations nearly cancel. The function-space check therefore complements the coefficient-space check.

The residual check follows the normwise backward-error viewpoint \citep[Chapter~7]{higham2002accuracy}. The following identity gives the smallest matrix correction.
\begin{lemma}[Residual as a matrix correction]
\label{lem:nullspace_backward_error}
Let $A\ne0$ and $Q^\top Q=I$. Then
\begin{equation}
\label{eq:nullspace_backward_error}
    \min_{\Delta A:\,(A+\Delta A)Q=0}
    \frac{\|\Delta A\|_2}{\|A\|_2}
    =\frac{\|AQ\|_2}{\|A\|_2}.
\end{equation}
\end{lemma}
\begin{proof}
Any feasible correction satisfies $\Delta A Q=-AQ$, so $\|\Delta A\|_2\ge\|AQ\|_2$. The choice $\Delta A=-AQQ^\top$ attains equality.
\end{proof}

For $A=\bar C_0$ or $A=\bar C_v$, this identity gives the smallest relative matrix correction that annihilates the candidate space. The correction is unstructured. It does not bound the unknown PDE or derivative error. The checks therefore provide numerical evidence, not a proof of symmetry.

Equation~\eqref{eq:validated_dimension} selects the largest space that passes all four checks. Minimizing residual alone would favor smaller nested spaces. We test every candidate because stability need not improve monotonically with dimension. We use the same empirical settings, $\varepsilon_R=10^{-3}$ and $\varepsilon_P=\varepsilon_X=10^{-2}$, for all benchmarks. The residual tolerance permits a relative matrix correction of at most $0.1\%$ in Lemma~\ref{lem:nullspace_backward_error}. Each stability tolerance limits the largest principal angle to about $0.57^\circ$. These are fixed acceptance tolerances. The cited results justify the measures, but do not prescribe these numerical values. The validation set participates in selection. It is not an independent test of the final discovery result.

For a selected dimension $d$, we return $Q=Q_{1,d}$, the already validated basis from the first fixed subset. The condition number is
\begin{equation}
\label{eq:validated_condition_number}
    \kappa=\max_{b=1,2,3}\frac{\sigma_{b,1}}{\sigma_{b,N_c-d}}.
\end{equation}
This uses the retained non-null spectrum without the numerical floor. A zero denominator gives $+\infty$. The caller retains the scales for pruning and coefficient recovery. The fourth output contains $\mathbf e=(e_0(d),e_R(d),e_P(d),e_X(d))$. Any reported spectral cutoff is only diagnostic; it does not change the selected basis.

If no candidate passes, we return $(\bot,+\infty,\emptyset,\emptyset)$. This denotes an unresolved evaluation, not a zero-dimensional null space. A trial library is then rejected. An unresolved initial library stops the search. If all five constraint matrices vanish, the full-space candidate must still pass the function rank and stability checks. We then return $d=N_c$, $\kappa=1$, and $Q=V_1$. A failed check gives an unresolved result. The result is also unresolved if only some constraint matrices vanish. Algorithm~\ref{alg:adaptive_svd_filtering} gives the procedure. 

\begin{algorithm}[!t]
\caption{Adaptive SVD}
\label{alg:adaptive_svd_filtering}
\compactalgorithm
\begin{algorithmic}[1]
\REQUIRE Normalized blocks $\bar B$, generator-value matrix $H$, cached Jacobians $J$, fixed sets $\mathcal I$, floor $\delta=10^{-12}$, tolerances $\varepsilon_R=10^{-3}$ and $\varepsilon_P=\varepsilon_X=10^{-2}$
\ENSURE Dimension $d$, condition number $\kappa$, normalized basis $Q$, diagnostics $\mathbf e$; or an unresolved result

\STATE Stack $J_i\bar B_i$ into $\bar C_b$ on $\mathcal I_b$ for $b=0,1,2,3$ and into $\bar C_v$ on $\mathcal I_v$; set $N_c\leftarrow\operatorname{cols}(\bar C_v)$.
\FOR{$b=1,2,3$}
    \STATE Decompose $\bar G_b=\bar C_b^\top\bar C_b=V_b\Sigma_b^2V_b^\top$.
    \STATE Set $\sigma_{b,j}=\sqrt{\max(\Sigma_{b,j}^2,0)}$ and sort the spectrum and $V_b$ in descending order.
    \STATE If $\sigma_{b,1}>0$, find $\mathcal J_b$ by Eq.~\eqref{eq:dimension_gap_candidates}; otherwise set $\mathcal J_b\leftarrow\emptyset$.
\ENDFOR
\STATE If all five matrices vanish, apply the function checks to $Q_{b,N_c}=V_b$. Return $(N_c,1,V_1,(0,0,0,e_X(N_c)))$ if they pass; otherwise return an unresolved result.
\STATE If any of the five matrices vanishes, \RETURN $(\bot,+\infty,\emptyset,\emptyset)$.
\STATE Form $\mathcal N$ by Eq.~\eqref{eq:dimension_gap_candidates}; set $\mathcal N_a\leftarrow\emptyset$.
\FOR{each $d\in\mathcal N$}
    \STATE Set $Q_{b,d}\leftarrow V_b[:,N_c-d+1:N_c]$ for all $b$.
    \STATE Compute $(e_0(d),e_R(d),e_P(d))$ by Eq.~\eqref{eq:dimension_validation}; compute $e_X(d)$ with the rank check in Eqs.~\eqref{eq:generator_function_rank}--\eqref{eq:generator_function_stability}.
    \STATE If all four checks in Eq.~\eqref{eq:validated_dimension} pass, add $d$ to $\mathcal N_a$.
\ENDFOR
\STATE If $\mathcal N_a=\emptyset$, \RETURN $(\bot,+\infty,\emptyset,\emptyset)$.
\STATE Set $d\leftarrow\max\mathcal N_a$, $Q\leftarrow Q_{1,d}$, and compute $\kappa$ by Eq.~\eqref{eq:validated_condition_number}.
\RETURN $(d,\allowbreak\kappa,\allowbreak Q,\allowbreak(e_0(d),e_R(d),e_P(d),e_X(d)))$
\end{algorithmic}
\end{algorithm}

\subsection{Basis Sparsification and Pruning}

Before sparsification, we remove symbolically verified linear identities in the library. For column vectors of features, an exact reduction takes the form
\begin{equation}
\label{eq:library_identity_reduction}
\Theta_{\mathrm{o}}=T\Theta_{\mathrm{r}},\qquad W_{a,\mathrm{r}}=W_{a,\mathrm{o}}T.
\end{equation}
Thus each generator function is unchanged. We restore the original coefficients before this map and extract independent mapped directions in the original coordinates using a machine-precision rank tolerance. We then apply the retained column scales and orthonormalize in normalized coordinates. Approximate dependence on sampled data is not treated as an identity. The reduced system must pass the four dimension checks, and its recovered span must agree with the mapped span within $\varepsilon_U$. We also check the mapped basis on the scale and validation pools using $\varepsilon_R$. This procedure, denoted by $\mathrm{ReduceIdentities}$, updates $(\Theta^\ast,S^\ast,Q)$. If a mandatory term is eliminated, we protect the retained terms in its exact representation. The pruning reference dimension is then $d^\ast=\operatorname{cols}(Q)$, rather than the unreduced coefficient nullity.

The resulting coefficient vectors are not necessarily sparse. We use the Linearized Alternating Direction Method with Adaptive Penalty (LADMAP)~\citep{lin2011linearized} to solve Eq.~\eqref{eq:LADMAP} in these normalized coordinates. The rotation preserves the span of \(Q\) while encouraging sparse entries in \(QP\).

We compute support scores from the normalized coefficient matrices using Eq.~\eqref{eq:normalized_support}. Low-support terms are tested one by one, and each retained coefficient column keeps its original scale. Every trial and final library must pass the same four dimension checks in adaptive SVD. Equal numerical nullity alone does not ensure preservation of the recovered generator subspace. We therefore restore the reference basis $Q$ and each trial basis $Q_{\mathrm{trial}}$ using $\operatorname{vec}(W)=S^{-1}\operatorname{vec}(\bar W)$. Let $B_{\mathrm{ref}}$ and $B_{\mathrm{trial}}$ stack the restored coefficient vectors as columns. The embedding $E_{\mathrm{trial}}$ inserts zeros for removed expressions within each variable block, following the order of $\Theta^\ast$. We form
\begin{equation}
\label{eq:pruning_subspace_bases}
\begin{aligned}
    U_{\mathrm{ref}}&=\operatorname{orth}(B_{\mathrm{ref}}),\\
    U_{\mathrm{trial}}&=\operatorname{orth}(E_{\mathrm{trial}}B_{\mathrm{trial}}),
\end{aligned}
\end{equation}
where $\operatorname{orth}$ denotes an orthonormal column-space basis obtained by a thin QR factorization. The distance check uses the following standard identity \citep[Lemma~1.5]{drineas2019low}.

\begin{lemma}[Projector distance]
\label{lem:pruning_subspace_distance}
Let $U_{\mathrm{ref}}$ and $U_{\mathrm{trial}}$ be orthonormal bases of two $d$-dimensional subspaces in the same coefficient space, with $d>0$. Then
\begin{equation}
\label{eq:pruning_subspace_distance}
\begin{aligned}
    e_{\mathrm{trial}}
    &=\left\|U_{\mathrm{ref}}U_{\mathrm{ref}}^\top
    -U_{\mathrm{trial}}U_{\mathrm{trial}}^\top\right\|_2\\
    &=\sin\theta_{\max},
\end{aligned}
\end{equation}
where $\theta_{\max}$ is the largest principal angle between the two subspaces.
\end{lemma}

For all trials, we define $e_{\mathrm{trial}}$ by the projector norm in Eq.~\eqref{eq:pruning_subspace_distance}. We remove a term only if
\begin{equation}
\label{eq:pruning_acceptance}
\begin{gathered}
    d_{\mathrm{trial}}=d_{\mathrm{ref}},\quad
    \kappa_{\mathrm{trial}}\le\gamma\kappa_{\mathrm{ref}},\quad e_{\mathrm{trial}}\le\varepsilon_U,
\end{gathered}
\end{equation}
where $0<\varepsilon_U<1$ is a prescribed tolerance. The reference space of $\Theta^\ast$ remains fixed throughout pruning, limiting cumulative drift. If both spaces have dimension zero, we set $e_{\mathrm{trial}}=0$. The QR bases in Eq.~\eqref{eq:pruning_subspace_bases} are used only for comparison and do not replace the sparse output basis. We then recompute and sparsify the final normalized basis and restore the original coefficients using Eq.~\eqref{eq:pruned_coefficient_recovery}. Algorithm~\ref{alg:sparsification_pruning} summarizes these steps.

\begin{algorithm}[!htbp]
\caption{Basis Sparsification and Pruning}
\label{alg:sparsification_pruning}
\compactalgorithm
\begin{algorithmic}[1]
\REQUIRE Orthonormal null basis $Q$, library $\Theta^\ast$, scales $S^\ast$, mandatory terms $\Theta_0$, cached Jacobians $J$, fixed sets $\mathcal I$, tolerances $\gamma\ge1$ and $0<\varepsilon_U<1$
\ENSURE Pruned library $\Theta_f$, sparse coefficients $Q_{\mathrm{final}}^\ast$ in the original coordinates

\STATE Apply $\mathrm{ReduceIdentities}$ to $(\Theta^\ast,S^\ast,Q)$ using $J$ and $\mathcal I$, as described above.
\STATE Solve Eq.~\eqref{eq:LADMAP} by LADMAP for $P$; set $Z\leftarrow QP$.
\STATE Reshape $Z$ into $\{\bar W_i\}$ and compute $s_j$ using Eq.~\eqref{eq:normalized_support}.
\STATE Sort non-mandatory terms $\Theta^\ast\setminus\Theta_0$ by ascending $s_j$.
\STATE Set $\Theta_f\leftarrow\Theta^\ast$, $S_f\leftarrow S^\ast$.
\STATE $(\bar B_f,H_f,S_f)\leftarrow\mathrm{Normalize}(\Theta_f;S_f,\mathcal I)$.
\STATE $(d_{\mathrm{ref}},\allowbreak\kappa_{\mathrm{ref}},\cdot,\cdot)\leftarrow\mathrm{AdaptiveSVD}(\bar B_f,H_f,J,\mathcal I)$.
\STATE Restore the mapped basis $Q$ using $S^\ast$ and fix $U_{\mathrm{ref}}$ by Eq.~\eqref{eq:pruning_subspace_bases}.

\FOR{each candidate term $\theta_j$ in sorted order}
    \STATE $\Theta_{\mathrm{trial}}\leftarrow\Theta_f\setminus\{\theta_j\}$.
    \STATE Restrict $S_f$ to the retained terms in every variable block as $S_{\mathrm{trial}}$.
    \STATE $(\bar B_{\mathrm{trial}},H_{\mathrm{trial}},S_{\mathrm{trial}})\leftarrow\mathrm{Normalize}(\Theta_{\mathrm{trial}};S_{\mathrm{trial}},\mathcal I)$.
    \STATE $(d_{\mathrm{trial}},\allowbreak\kappa_{\mathrm{trial}},Q_{\mathrm{trial}},\cdot)\leftarrow\mathrm{AdaptiveSVD}(\bar B_{\mathrm{trial}},H_{\mathrm{trial}},J,\mathcal I)$; skip this trial if $d_{\mathrm{trial}}=\bot$.
    \STATE Restore $Q_{\mathrm{trial}}$ using $S_{\mathrm{trial}}$, embed it in $\Theta^\ast$, and compute $e_{\mathrm{trial}}$ by Eqs.~\eqref{eq:pruning_subspace_bases}--\eqref{eq:pruning_subspace_distance}.
    \IF{the conditions in Eq.~\eqref{eq:pruning_acceptance} hold}
        \STATE $(\Theta_f,S_f)\leftarrow(\Theta_{\mathrm{trial}},S_{\mathrm{trial}})$.
    \ENDIF
\ENDFOR

\STATE $(\bar B_f,H_f,S_f)\leftarrow\mathrm{Normalize}(\Theta_f;S_f,\mathcal I)$.
\STATE $(\cdot,\cdot,Q,\cdot)\leftarrow\mathrm{AdaptiveSVD}(\bar B_f,H_f,J,\mathcal I)$.
\STATE Solve Eq.~\eqref{eq:LADMAP} by LADMAP for $P_f$; set $Z_f\leftarrow QP_f$.
\STATE Reshape the columns of $Z_f$ into normalized coefficient matrices $\{\bar W_i\}$.
\STATE Recover $\operatorname{vec}(W_i)\leftarrow S_f^{-1}\operatorname{vec}(\bar W_i)$; stack these vectors as $Q_{\mathrm{final}}^\ast$.
\RETURN $\Theta_f,\allowbreak Q_{\mathrm{final}}^\ast$
\end{algorithmic}
\end{algorithm}

\section{Experimental Settings}
\label{app:settings}

\subsection{Computing Infrastructure}

All experiments were conducted on a server equipped with an Intel Xeon Gold 6330 CPU running at 2.00~GHz, 1.5~TB of RAM, and NVIDIA A100 GPUs with 80~GB of memory each.

Regarding computational cost, the generator discovery stage of {LieDiscover} is lightweight, requiring only approximately 600~MiB of GPU memory during training.

\subsection{Hyper-parameters}
All hyper-parameters used in our experiments are documented in the released source code. The implementation includes the corresponding configuration files and parameter settings for each experiment, facilitating reproducibility.

\begin{figure*}[pos=tb]
    \centering
    \includegraphics[width=\linewidth]{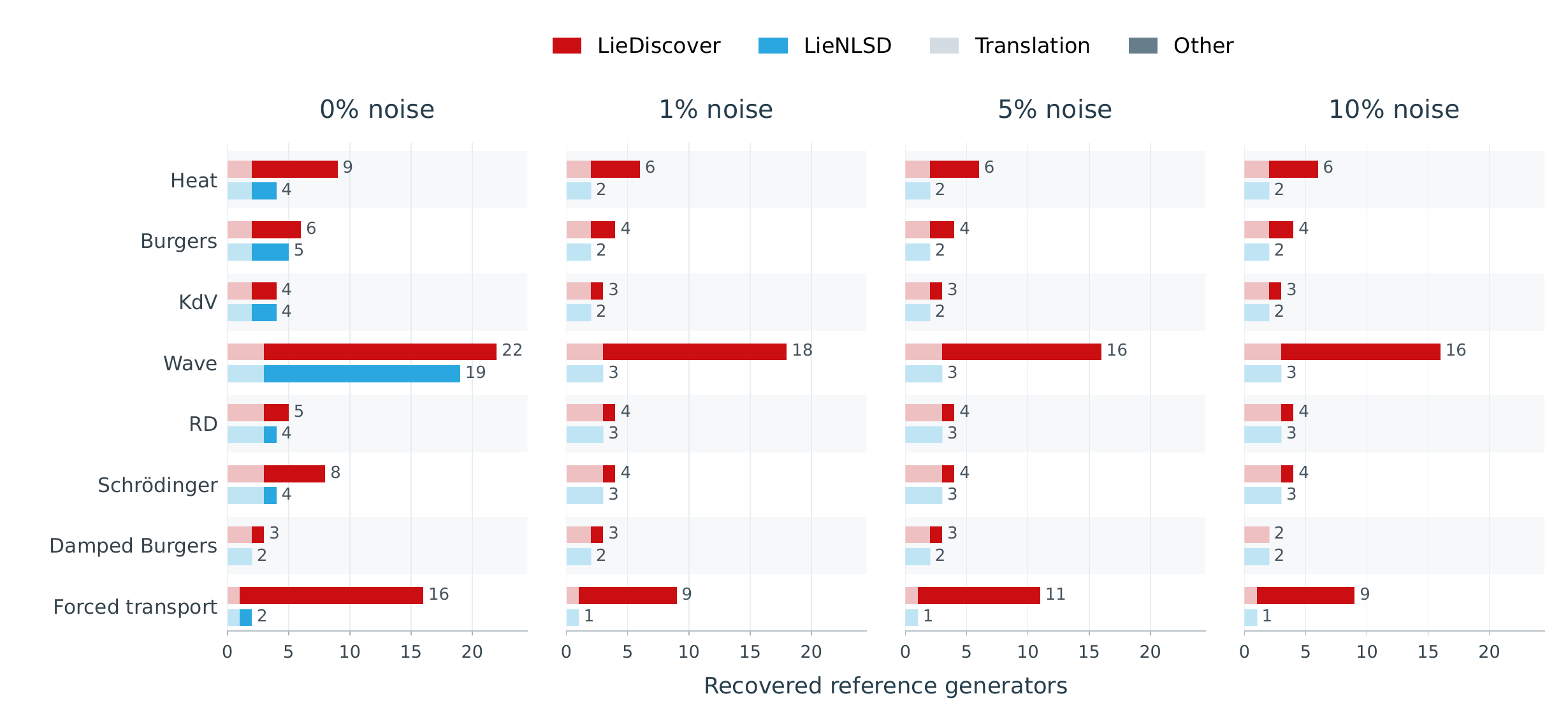}
    \caption{Recovered reference directions on eight PDEs. Red bars denote LieDiscover, and blue bars denote LieNLSD. Light segments count translations, and dark segments count the other recovered directions. Bar-end labels give the total recovered count.}
    \label{fig:noise_recovery_counts}
\end{figure*}

\subsection{Symmetry Subspace Distances}
\label{app:grassmann_results}

\begin{table}[pos=t]
\centering
\footnotesize
\setlength{\tabcolsep}{3pt}
\caption{Grassmann distances for generators recovered only by LieDiscover (mean $\pm$ standard deviation).}
\label{tab:grassmann_exclusive}
\begin{tabular}{lc}
\toprule
Dataset & LieDiscover \\
\midrule
Burgers' & $9.359\!\times\!10^{-5}\pm2.319\!\times\!10^{-5}$ \\
Heat & $1.462\!\times\!10^{-5}\pm6.768\!\times\!10^{-6}$ \\
Wave & $7.384\!\times\!10^{-8}\pm1.572\!\times\!10^{-8}$ \\
Schr\"odinger & $3.052\!\times\!10^{-2}\pm8.906\!\times\!10^{-3}$ \\
Damped Burgers' & $1.004\!\times\!10^{-5}\pm8.169\!\times\!10^{-8}$ \\
Forced transport & $5.188\!\times\!10^{-7}\pm7.133\!\times\!10^{-8}$ \\
\bottomrule
\end{tabular}
\end{table}

Tables~\ref{tab:grassmann_distance} and~\ref{tab:grassmann_shared} give the detailed results for the first two groups in Fig.~\ref{fig:grassmann_comparison}. Values are means $\pm$ standard deviations over three runs.

\begin{table*}[pos=tb]
\centering
\footnotesize
\setlength{\tabcolsep}{3pt}
\caption{Grassmann distances for homogeneous-linear symmetry subspaces (mean $\pm$ standard deviation). Lower is better, and the best result in each row is bold.}
\label{tab:grassmann_distance}
\begin{tabular}{lccc}
\toprule
Dataset & LieGAN & LieNLSD & LieDiscover \\
\midrule
Burgers' & $1.584\!\times\!10^{0}\pm1.993\!\times\!10^{-2}$ & $6.253\!\times\!10^{-2}\pm1.207\!\times\!10^{-2}$ & $\mathbf{8.130\!\times\!10^{-4}\pm5.667\!\times\!10^{-4}}$ \\
Heat & $2.245\!\times\!10^{0}\pm7.187\!\times\!10^{-2}$ & $7.050\!\times\!10^{-1}\pm6.221\!\times\!10^{-1}$ & $\mathbf{2.217\!\times\!10^{-5}\pm4.812\!\times\!10^{-6}}$ \\
KdV & $1.537\!\times\!10^{0}\pm1.405\!\times\!10^{-3}$ & $9.390\!\times\!10^{-3}\pm3.378\!\times\!10^{-3}$ & $\mathbf{4.078\!\times\!10^{-4}\pm1.628\!\times\!10^{-4}}$ \\
Wave & $2.600\!\times\!10^{0}\pm1.319\!\times\!10^{-1}$ & $4.014\!\times\!10^{-3}\pm1.861\!\times\!10^{-3}$ & $\mathbf{1.754\!\times\!10^{-7}\pm1.582\!\times\!10^{-8}}$ \\
Reaction--diffusion & $6.138\!\times\!10^{-1}\pm4.522\!\times\!10^{-2}$ & $1.649\!\times\!10^{-1}\pm2.284\!\times\!10^{-2}$ & $\mathbf{6.270\!\times\!10^{-4}\pm2.620\!\times\!10^{-4}}$ \\
Schr\"odinger & $2.504\!\times\!10^{0}\pm2.261\!\times\!10^{-1}$ & $2.318\!\times\!10^{-1}\pm3.278\!\times\!10^{-2}$ & $\mathbf{8.814\!\times\!10^{-3}\pm2.126\!\times\!10^{-3}}$ \\
Forced transport & $1.771\!\times\!10^{0}\pm9.834\!\times\!10^{-2}$ & $1.530\!\times\!10^{0}\pm3.060\!\times\!10^{-2}$ & $\mathbf{1.696\!\times\!10^{-8}\pm2.910\!\times\!10^{-9}}$ \\
\bottomrule
\end{tabular}
\end{table*}

\begin{table*}[pos=tb]
\centering
\footnotesize
\setlength{\tabcolsep}{3pt}
\caption{Grassmann distances for the remaining generators shared by LieNLSD and LieDiscover (mean $\pm$ standard deviation). Lower is better, and the best result in each row is bold.}
\label{tab:grassmann_shared}
\begin{tabular}{lcc}
\toprule
Dataset & LieNLSD & LieDiscover \\
\midrule
Burgers' & $1.156\!\times\!10^{-2}\pm1.998\!\times\!10^{-3}$ & $\mathbf{2.321\!\times\!10^{-6}\pm1.048\!\times\!10^{-6}}$ \\
Heat & $7.840\!\times\!10^{-3}\pm1.263\!\times\!10^{-3}$ & $\mathbf{5.516\!\times\!10^{-6}\pm6.269\!\times\!10^{-6}}$ \\
KdV & $2.680\!\times\!10^{-3}\pm7.121\!\times\!10^{-4}$ & $\mathbf{5.061\!\times\!10^{-6}\pm2.068\!\times\!10^{-11}}$ \\
Wave & $3.320\!\times\!10^{-2}\pm2.695\!\times\!10^{-2}$ & $\mathbf{4.745\!\times\!10^{-7}\pm6.683\!\times\!10^{-8}}$ \\
Schr\"odinger & $1.490\!\times\!10^{-8}\pm1.217\!\times\!10^{-8}$ & $\mathbf{9.934\!\times\!10^{-9}\pm1.405\!\times\!10^{-8}}$ \\
Forced transport & $2.715\!\times\!10^{-2}\pm1.903\!\times\!10^{-2}$ & $\mathbf{8.603\!\times\!10^{-9}\pm1.217\!\times\!10^{-8}}$ \\
\bottomrule
\end{tabular}
\end{table*}

\subsection{Held-Out Validation Results}
\label{app:held_out_results}
Table~\ref{tab:held_out_residuals} reports the training and held-out residuals defined in Eq.~\eqref{eq:heldout_rms} and visualized in Fig.~\ref{fig:heldout_comparison}. Each entry summarizes three runs evaluated on matched points for the two methods.

\begin{table*}[pos=tb]
\centering
\footnotesize
\setlength{\tabcolsep}{5pt}
\caption{Training and held-out symmetry residuals (mean $\pm$ population standard deviation). Training uses 100 fitting points and testing uses 1,000 independent points. Lower is better.}
\label{tab:held_out_residuals}
\begin{tabular}{llcc}
\toprule
Dataset & Method & Train & Test \\
\midrule
\multirow{2}{*}{Burgers'} & LieNLSD & $6.338\!\times\!10^{-2}\pm2.532\!\times\!10^{-2}$ & $4.105\!\times\!10^{-1}\pm1.097\!\times\!10^{-1}$ \\
 & LieDiscover & $4.075\!\times\!10^{-4}\pm3.031\!\times\!10^{-4}$ & $5.325\!\times\!10^{-3}\pm2.101\!\times\!10^{-3}$ \\
\addlinespace
\multirow{2}{*}{Heat} & LieNLSD & $1.105\!\times\!10^{-2}\pm1.417\!\times\!10^{-3}$ & $1.227\!\times\!10^{-2}\pm5.663\!\times\!10^{-3}$ \\
 & LieDiscover & $3.571\!\times\!10^{-5}\pm7.347\!\times\!10^{-6}$ & $6.419\!\times\!10^{-5}\pm1.407\!\times\!10^{-5}$ \\
\addlinespace
\multirow{2}{*}{KdV} & LieNLSD & $3.481\!\times\!10^{-2}\pm2.119\!\times\!10^{-2}$ & $3.830\!\times\!10^{-2}\pm5.199\!\times\!10^{-3}$ \\
 & LieDiscover & $1.002\!\times\!10^{-3}\pm3.183\!\times\!10^{-5}$ & $1.536\!\times\!10^{-3}\pm4.376\!\times\!10^{-5}$ \\
\addlinespace
\multirow{2}{*}{Damped Burgers'} & LieNLSD & $0.000\!\times\!10^{0}\pm0.000\!\times\!10^{0}$ & $0.000\!\times\!10^{0}\pm0.000\!\times\!10^{0}$ \\
 & LieDiscover & $1.966\!\times\!10^{-6}\pm5.896\!\times\!10^{-7}$ & $1.942\!\times\!10^{-6}\pm9.033\!\times\!10^{-10}$ \\
\addlinespace
\multirow{2}{*}{Forced transport} & LieNLSD & $1.709\!\times\!10^{-3}\pm1.909\!\times\!10^{-4}$ & $1.899\!\times\!10^{-3}\pm1.358\!\times\!10^{-4}$ \\
 & LieDiscover & $9.374\!\times\!10^{-8}\pm1.196\!\times\!10^{-8}$ & $9.906\!\times\!10^{-8}\pm1.065\!\times\!10^{-8}$ \\
\addlinespace
\multirow{2}{*}{Wave} & LieNLSD & $9.431\!\times\!10^{-3}\pm3.094\!\times\!10^{-3}$ & $1.085\!\times\!10^{-2}\pm2.774\!\times\!10^{-3}$ \\
 & LieDiscover & $7.792\!\times\!10^{-7}\pm9.164\!\times\!10^{-8}$ & $8.327\!\times\!10^{-7}\pm8.277\!\times\!10^{-8}$ \\
\addlinespace
\multirow{2}{*}{Reaction--diffusion} & LieNLSD & $5.725\!\times\!10^{-2}\pm1.565\!\times\!10^{-2}$ & $4.180\!\times\!10^{0}\pm8.598\!\times\!10^{-1}$ \\
 & LieDiscover & $1.872\!\times\!10^{-4}\pm8.825\!\times\!10^{-5}$ & $5.449\!\times\!10^{-3}\pm1.291\!\times\!10^{-3}$ \\
\addlinespace
\multirow{2}{*}{Schr\"odinger} & LieNLSD & $4.968\!\times\!10^{-1}\pm1.604\!\times\!10^{-1}$ & $5.279\!\times\!10^{-1}\pm2.414\!\times\!10^{-2}$ \\
 & LieDiscover & $1.133\!\times\!10^{-1}\pm1.256\!\times\!10^{-2}$ & $3.354\!\times\!10^{-1}\pm3.646\!\times\!10^{-2}$ \\
\bottomrule
\end{tabular}
\end{table*}

\subsection{Four-Way Residual Results and References}
\label{app:four_group_residuals}
Table~\ref{tab:four_group_residuals} gives the four residuals in Section~\ref{sec:four_group_residuals}. It reports the arithmetic mean of the per-run overall RMS scores, with population standard deviations for the discovered generators. The reference equation, generators, and ordered test points are fixed across seeds, so the reference columns contain one repeated deterministic evaluation rather than three independent estimates. 

The finite analytic reference sets contain $6$, $10$, $4$, $3$, $16$, $22$, $5$, and $9$ generators for Burgers', heat, KdV, damped Burgers', forced transport, wave, reaction--diffusion, and Schr\"odinger, respectively. For Burgers', KdV, damped Burgers', forced transport, wave, and reaction--diffusion, these are the analytic bases displayed in the corresponding equation descriptions and Table~\ref{tab:selected_generators}, up to sign and ordering. The heat reference augments the nine displayed directions with
\[
\mathbf v_{\mathrm{proj}}^{\mathrm{heat}}
=4t^2\partial_t+4tx\partial_x-(x^2+2t)u\partial_u.
\]
The Schr\"odinger reference augments its eight displayed directions with
\[
\begin{aligned}
\mathbf v_{\mathrm{proj}}^{\mathrm{Sch}}
={}&t^2\partial_t+tx\partial_x+ty\partial_y\\
&-\left(tu+\frac{x^2+y^2}{2}v\right)\partial_u\\
&+\left(-tv+\frac{x^2+y^2}{2}u\right)\partial_v.
\end{aligned}
\]
These fixed references do not assert complete recovery of the infinite-dimensional families present in heat, wave, forced transport, or potential Burgers'. Each reference vector is expanded in its recorded coefficient library and normalized separately, using the same coefficient-norm convention as the discovered vectors. Comparisons between the two equation columns within a generator set retain that normalization; comparisons across generator sets also depend on their representations and coverage.

\begin{table*}[pos=tb]
\centering
\footnotesize
\setlength{\tabcolsep}{3pt}
\caption{Four-way residuals on the same 1,000 held-out numerical jets, without added observation noise. Discovered-generator entries are means $\pm$ population standard deviations over three search seeds. The fixed reference evaluations are identical across these seeds and are listed once. $m_D$ and $m_R$ denote discovered and reference generator counts.}
\label{tab:four_group_residuals}
\begin{tabular}{lcccccc}
\toprule
 & & & \multicolumn{2}{c}{Discovered generators $\widehat{\mathcal V}$}
 & \multicolumn{2}{c}{Reference generators $\mathcal V^\star$} \\
\cmidrule(lr){4-5}\cmidrule(lr){6-7}
Dataset & $m_D$ & $m_R$ & WSINDy $\widehat F$ & True $F$ & WSINDy $\widehat F$ & True $F$ \\
\midrule
Burgers' & 6 & 6 & $5.203\!\times\!10^{-3}\pm2.114\!\times\!10^{-3}$ & $5.325\!\times\!10^{-3}\pm2.101\!\times\!10^{-3}$ & $1.566\!\times\!10^{-3}$ & $1.613\!\times\!10^{-3}$ \\
Heat & 9 & 10 & $6.473\!\times\!10^{-5}\pm1.395\!\times\!10^{-5}$ & $6.419\!\times\!10^{-5}\pm1.407\!\times\!10^{-5}$ & $1.417\!\times\!10^{-4}$ & $1.486\!\times\!10^{-4}$ \\
KdV & 4 & 4 & $1.539\!\times\!10^{-3}\pm3.184\!\times\!10^{-5}$ & $1.536\!\times\!10^{-3}\pm4.376\!\times\!10^{-5}$ & $2.552\!\times\!10^{-3}$ & $2.580\!\times\!10^{-3}$ \\
Damped Burgers' & 3 & 3 & $2.394\!\times\!10^{-8}\pm4.017\!\times\!10^{-9}$ & $1.942\!\times\!10^{-6}\pm9.033\!\times\!10^{-10}$ & $1.946\!\times\!10^{-6}$ & $1.837\!\times\!10^{-17}$ \\
Forced transport & 16/15/15 & 16 & $1.161\!\times\!10^{-7}\pm1.855\!\times\!10^{-9}$ & $9.906\!\times\!10^{-8}\pm1.065\!\times\!10^{-8}$ & $1.484\!\times\!10^{-7}$ & $5.784\!\times\!10^{-16}$ \\
Wave & 22 & 22 & $2.329\!\times\!10^{-7}\pm3.233\!\times\!10^{-8}$ & $8.327\!\times\!10^{-7}\pm8.277\!\times\!10^{-8}$ & $6.559\!\times\!10^{-7}$ & $7.040\!\times\!10^{-7}$ \\
Reaction--diffusion & 5 & 5 & $5.364\!\times\!10^{-3}\pm1.293\!\times\!10^{-3}$ & $5.449\!\times\!10^{-3}\pm1.291\!\times\!10^{-3}$ & $4.668\!\times\!10^{-3}$ & $4.758\!\times\!10^{-3}$ \\
Schr\"odinger & 8 & 9 & $3.223\!\times\!10^{-1}\pm3.775\!\times\!10^{-2}$ & $3.354\!\times\!10^{-1}\pm3.646\!\times\!10^{-2}$ & $5.256\!\times\!10^{-1}$ & $6.131\!\times\!10^{-1}$ \\
\bottomrule
\end{tabular}
\end{table*}

\subsection{Symmetry Recovery under Observation Noise}
\label{app:noise_recovery}

Fig.~\ref{fig:noise_recovery_counts} gives the recovered reference counts for the experiment in Section~\ref{sec:noise_robustness}. The reference sets contain $10$, $6$, $4$, $22$, $5$, $9$, $3$, and $16$ directions for heat, Burgers', KdV, wave, reaction--diffusion, Schr\"odinger, damped Burgers', and forced transport, respectively.

The additional directions account for most of LieDiscover's advantage under noise. At $1\%$ noise, it recovers $15$ non-translation directions on wave, $8$ on forced transport, and $4$ on heat. Reaction--diffusion and Schr\"odinger each retain three translations and spatial rotation at all three noise levels.

\subsection{Search Ablation Results}
\label{app:search_ablation}

Table~\ref{tab:search_ablation_all} reports the library budgets and recovered counts for all eight PDEs. Here $D_0=\nu(\Theta_0)$ is the true dimension supported by the initial library. $D_K$ is the capacity with at most $K$ added functions, and $D_{\mathrm{full}}$ is the complete-pool capacity. The recovery ratio divides each count by $D_{\mathrm{full}}$ for Fixed library and by $D_K$ for Monte Carlo and LieDiscover.

The capacities are certified by symbolic determining equations. For polynomial libraries, we match polynomial coefficients and compute exact null spaces. Forced transport uses Fourier coefficients, and damped Burgers' uses its complete point-symmetry algebra. The budgeted maxima combine upper bounds over library subsets with explicit libraries that attain them. Function identities are removed when counting independent generators.

\begin{table*}[pos=tb]
\centering
\small
\setlength{\tabcolsep}{5pt}
\caption{Search budgets, representable capacities, and validated directions per run. Counts are means $\pm$ population standard deviations over three seeds. Runs with no accepted output contribute zero recovered directions.}
\label{tab:search_ablation_all}
\begin{tabular}{lrrrrccc}
\toprule
& & & & & \multicolumn{3}{c}{Validated directions} \\
Dataset & $K$ & $D_0$ & $D_K$ & $D_{\mathrm{full}}$ & Fixed library & Monte Carlo & LieDiscover \\
\midrule
Burgers'          & 6  & 5  & 6  & 6  & $3.00\pm0.00$ & $6.00\pm0.00$  & $6.00\pm0.00$ \\
Heat              & 6  & 6  & 9  & 12 & $4.00\pm0.00$ & $9.00\pm0.00$  & $9.00\pm0.00$ \\
KdV               & 6  & 4  & 4  & 4  & $2.00\pm0.00$ & $4.00\pm0.00$  & $4.00\pm0.00$ \\
Wave              & 15 & 12 & 24 & 45 & $0.00\pm0.00$ & $18.33\pm0.47$ & $22.00\pm0.00$ \\
Reaction--diffusion & 15 & 5 & 5 & 5  & $0.00\pm0.00$ & $5.00\pm0.00$  & $5.00\pm0.00$ \\
Schr\"odinger      & 24 & 6  & 9  & 9  & $3.00\pm0.00$ & $7.33\pm0.47$  & $8.00\pm0.00$ \\
Damped Burgers'   & 6  & 2  & 3  & 3  & $0.00\pm0.00$ & $3.00\pm0.00$  & $3.00\pm0.00$ \\
Forced transport  & 8  & 4  & 17 & 71 & $3.33\pm4.71$ & $11.33\pm1.25$ & $15.33\pm0.47$ \\
\bottomrule
\end{tabular}
\end{table*}

Fixed library returns an accepted output in one of the three forced-transport runs. It returns no accepted output on wave, reaction--diffusion, or damped Burgers'. These runs enter the recovery score with zero recovered directions. KdV and reaction--diffusion have $D_0=D_K$, so their results are reported here alongside the six expandable cases in Fig.~\ref{fig:search_ablation}.

\subsection{Baseline Settings}
\subsubsection{Symmetry Discovery}
For the symmetry discovery task, the primary baselines are LieGAN and LieNLSD. Their experimental settings are as follows.

\noindent\textbf{LieGAN.}
LieGAN~\citep{yang2023generative} learns infinitesimal generators through adversarial training. It uses Lie algebra generators to transform data and trains a discriminator to distinguish the transformed distribution from the original one. For Burgers' equation, Heat equation, KdV equation, Wave equation, Reaction-Diffusion equation, and Schr\"odinger equation, we follow the same data splits and hyperparameter settings used in the LieNLSD \citep{hu2025explicit} comparison. For the two newly introduced datasets, Damped Burgers' equation and Forced transport equation, we train LieGAN on exactly the same generated training data as our method. We use the original LieGAN implementation with batch size \(64\), \(100\) training epochs, discriminator learning rate \(2\times 10^{-4}\), generator learning rate \(10^{-3}\), regularization weight \(1\), and random seeds \(\{0,1,2\}\).

\noindent\textbf{LieNLSD.}
LieNLSD~\citep{hu2025explicit} represents candidate symmetries using a predefined function library and uses SVD to identify the number and explicit forms of infinitesimal generators. For all datasets already included in LieNLSD, namely Burgers' equation, Heat equation, KdV equation, Wave equation, Reaction-Diffusion equation, and Schr\"odinger equation, we keep the dataset, surrogate training, sampling strategy, and candidate library setting consistent with the original LieNLSD experiments. For the two newly introduced PDEs, we use the surrogate architecture, an MLP with three hidden layers, hidden dimension \(200\), Sigmoid activation, and \(200\) training epochs. The sample window used for symmetry discovery is fixed across methods by seed.

\bibliographystyle{cas-model2-names}
\bibliography{liediscover}

\end{document}